\documentclass[onecolumn]{IEEEtran}
\usepackage[latin1]{inputenc}
\usepackage[english]{babel}
\usepackage[margin=1in]{geometry}

\usepackage{blindtext}
\usepackage{amssymb}
\usepackage{hyperref}
\hypersetup{
    colorlinks=true,
    linkcolor=blue,
    filecolor=magenta,
    urlcolor=cyan,
}

\usepackage{amsthm}
\usepackage[normalem]{ulem}
\usepackage{amsmath}
\usepackage{booktabs}
\usepackage{array}
\usepackage{cite}
\usepackage{soul}

\usepackage{comment}
\usepackage{xcolor}

\newtheorem{theorem}{Theorem}[section]
\newtheorem{corollary}[theorem]{Corollary}
\newtheorem{definition}[theorem]{Definition}
\newtheorem{lemma}[theorem]{Lemma}

\newtheorem{claim}[theorem]{Claim}
\newtheorem{remark}[theorem]{Remark}

\newtheorem{construction}{Construction}
\newtheorem{example}{Example}

\newcommand\nc\newcommand
\nc{\cA}{\mathcal{A}}\nc{\cB}{\mathcal{B}}\nc{\cC}{\mathcal{C}}\nc{\cD}{\mathcal{D}}
\nc{\cE}{\mathcal{E}}\nc{\cF}{\mathcal{F}}\nc{\cG}{\mathcal{G}}\nc{\cH}{\mathcal{H}}
\nc{\cI}{\mathcal{I}}\nc{\cJ}{\mathcal{J}}\nc{\cK}{\mathcal{K}}\nc{\cL}{\mathcal{L}}
\nc{\cM}{\mathcal{M}}\nc{\cN}{\mathcal{N}}\nc{\cO}{\mathcal{O}}\nc{\cP}{\mathcal{P}}
\nc{\cQ}{\mathcal{Q}}\nc{\cR}{\mathcal{R}}\nc{\cS}{\mathcal{S}}\nc{\cT}{\mathcal{T}}
\nc{\cU}{\mathcal{U}}\nc{\cV}{\mathcal{V}}\nc{\cW}{\mathcal{W}}\nc{\cX}{\mathcal{X}}
\nc{\cY}{\mathcal{Y}}\nc{\cZ}{\mathcal{Z}}

\nc{\bba}{\mathbf{a}}\nc{\bbb}{\mathbf{b}}\nc{\bbc}{\mathbf{c}}\nc{\bbd}{\mathbf{d}}
\nc{\bbe}{\mathbf{e}}\nc{\bbf}{\mathbf{f}}\nc{\bbg}{\mathbf{g}}\nc{\bbh}{\mathbf{h}}
\nc{\bbi}{\mathbf{i}}\nc{\bbj}{\mathbf{j}}\nc{\bbk}{\mathbf{k}}\nc{\bbl}{\mathbf{l}}
\nc{\bbm}{\mathbf{m}}\nc{\bbn}{\mathbf{n}}\nc{\bbo}{\mathbf{o}}\nc{\bbp}{\mathbf{p}}
\nc{\bbq}{\mathbf{q}}\nc{\bbr}{\mathbf{r}}\nc{\bbs}{\mathbf{s}}\nc{\bbt}{\mathbf{t}}
\nc{\bbu}{\mathbf{u}}\nc{\bbv}{\mathbf{v}}\nc{\bbw}{\mathbf{w}}\nc{\bbx}{\mathbf{x}}
\nc{\bby}{\mathbf{y}}\nc{\bbz}{\mathbf{z}}

\nc{\bbA}{\mathbf{A}}\nc{\bbB}{\mathbf{B}}\nc{\bbC}{\mathbf{C}}\nc{\bbD}{\mathbf{D}}
\nc{\bbE}{\mathbf{E}}\nc{\bbF}{\mathbf{F}}\nc{\bbG}{\mathbf{G}}\nc{\bbH}{\mathbf{H}}
\nc{\bbI}{\mathbf{I}}\nc{\bbJ}{\mathbf{J}}\nc{\bbK}{\mathbf{K}}\nc{\bbL}{\mathbf{L}}
\nc{\bbM}{\mathbf{M}}\nc{\bbN}{\mathbf{N}}\nc{\bbO}{\mathbf{O}}\nc{\bbP}{\mathbf{P}}
\nc{\bbQ}{\mathbf{Q}}\nc{\bbR}{\mathbf{R}}\nc{\bbS}{\mathbf{S}}\nc{\bbT}{\mathbf{T}}
\nc{\bbU}{\mathbf{U}}\nc{\bbV}{\mathbf{V}}\nc{\bbW}{\mathbf{W}}\nc{\bfX}{\mathbf{X}}
\nc{\bbY}{\mathbf{Y}}\nc{\bbZ}{\mathbf{Z}}

\nc{\sA}{\mathsf{A}}\nc{\sB}{\mathsf{B}}\nc{\sC}{\mathsf{C}}\nc{\sD}{\mathsf{D}}
\nc{\sE}{\mathsf{E}}\nc{\sF}{\mathsf{F}}\nc{\sG}{\mathsf{G}}\nc{\sH}{\mathsf{H}}
\nc{\sI}{\mathsf{I}}\nc{\sJ}{\mathsf{J}}\nc{\sK}{\mathsf{K}}\nc{\sL}{\mathsf{L}}
\nc{\sM}{\mathsf{M}}\nc{\sN}{\mathsf{N}}\nc{\sO}{\mathsf{O}}\nc{\sP}{\mathsf{P}}
\nc{\sQ}{\mathsf{Q}}\nc{\sR}{\mathsf{R}}\nc{\sS}{\mathsf{S}}\nc{\sT}{\mathsf{T}}
\nc{\sU}{\mathsf{U}}\nc{\sV}{\mathsf{V}}\nc{\sW}{\mathsf{W}}\nc{\sX}{\mathsf{X}}
\nc{\sY}{\mathsf{Y}}\nc{\sZ}{\mathsf{Z}}

\newcommand{\abs}[1]{\left|#1\right|}

\newcommand{\parenv}[1]{\left( #1 \right)}
\newcommand{\sparenv}[1]{\left[ #1 \right]}

\nc{\set}[1]{\llbracket #1 \rrbracket}

\newcommand{\bal}[1]{\begin{align}\label{#1}}
\newcommand{\eal}{\end{align}}

\renewcommand{\leq}{\leqslant}

\renewcommand{\geq}{\geqslant}

\renewcommand{\Bbb}{\mathbb}

\renewcommand{\Bbb}{\mathbb}

\newcommand{\Fq}{{{\Bbb F}}_{\!q}}

\newcommand{\E}{{\Bbb E}}

\newcommand{\bcom}[1]{{\color{blue} #1 }}

\newcommand{\ms}[1]{\{\{ #1 \}\}}

\title{Batch Array Codes}
\author{Xiangliang~Kong, Chen~Wang, and Yiwei~Zhang%
\thanks{Research supported in part by the National Key Research and Development Program of China under Grant 2022YFA1004900 and Grant 2021YFA1001000, in part by the National Natural Science Foundation of China under Grant 12231014, in part by the Shandong Provincial Natural Science Foundation under Grant No. ZR2021YQ46, in part by the Taishan Scholar Program of Shandong Province, and in part by the European Research Council (ERC) under Grant 852953.}
\thanks{X. Kong (Corresponding author: rongxlkong@gmail.com) is with the Department of Electrical Engineering-Systems, Tel Aviv University, Tel Aviv-Yafo 6997801, Israel. C. Wang and Y. Zhang are with the Key Laboratory of Cryptologic Technology and Information Security, Ministry of Education, and also with the School of Cyber Science and Technology, Shandong University, Qingdao, Shandong 266237, China (e-mails: cwang2021@mail.sdu.edu.cn, ywzhang@sdu.edu.cn).}
}

\begin{document}

\maketitle

\begin{abstract}
    Batch codes are a type of codes specifically designed for coded distributed storage systems and private information retrieval protocols. These codes have got much attention in recent years due to their ability to enable efficient and secure storage in distributed systems.

    In this paper, we study an array code version of the batch codes, which is called the \emph{batch array code} (BAC). Under the setting of BAC, each node stores a bucket containing multiple code symbols and responds with a locally computed linear combination of the symbols in its bucket during the recovery of a requested symbol. We demonstrate that BACs can support the same type of requests as the original batch codes but with reduced redundancy. Specifically, we establish information theoretic lower bounds on the code lengths and provide several code constructions that confirm the tightness of the lower bounds for certain parameter regimes.
\end{abstract}

\section{Introduction}

As an important class of codes designed for load balancing in distributed storage systems and applicable in cryptographic protocols such as private information retrieval (PIR), multi-set batch codes were first introduced by Ishai et al. \cite{IKOS04}.
In their original setup, an \emph{$(n,N,k,m,t)$ multi-set batch code} (\emph{batch code} for short) encodes $n$ information symbols, $(x_1,\ldots, x_n)\in \Sigma^n$, into $N$ code symbols.
These code symbols are divided as $m$ buckets $(\bbc_1,\ldots,\bbc_m)\in (\Sigma^{\geq1})^m$ and distributed across $m$ separate storage nodes, such that for any multi-set of $k$ indices $\ms{i_1,\ldots,i_k}$, the buckets can be divided into $k$ disjoint recovery sets and $x_{i_j}$ can be retrieved by probing at most $t$ code symbols from every bucket in the $j$-th recovery set. The multi-set of indices can be seen as a $k$-batch request from a group of $k$ users, one request from each user.
Each user has access to only one specific recovery set and each recovery set only serves for one particular user.
The disjointness of the recovery sets and the role of the parameter $t$ together aim for balancing the access loads among the servers.
In their initial work, Ishai et al. \cite{IKOS04} used unbalanced expanders, subcube codes, smooth codes, and Reed-Muller codes to obtain batch codes for a constant rate $\rho\triangleq\frac{n}{N}<1$. Then, using graph theory, Rawat et al. \cite{RSDG16} constructed batch codes that achieve asymptotically optimal rates of $1-o_t(1)$.

As the definition suggests, batch codes can handle multiple requests from different users simultaneously, a property referred to as \emph{availability}.
The availability property is critical for achieving high throughput in distributed storage systems and has been extensively studied for other storage codes, such as locally repairable codes (LRCs), as seen in \cite{TB14,WZ14,HYUS15,RPDV16,TBF16,CMST20}. Moreover, as is shown in \cite{IKOS04,ACLS18}, batch codes can also be used to amortize the computational costs of private information retrieval (PIR) and other related cryptographic protocols. Recently, in the study of service rate region of codes in distributed systems, Akta\c{s} et al. \cite{AJKKS21} also proved that batch codes can be used to maximize the integral service rate region.

Over the years, there have been many works exploring different variants of batch codes.
One such variant, known as \emph{combinatorial batch codes}, arises when the symbols stored in the buckets are simply copies of the information symbols.
These codes have been extensively studied, for examples, see \cite{PSW09,BRR12,SG16,ST20,CKZ21}.
A special class of batch codes with $t=n$, called \emph{switch codes}, has been explored in \cite{WSCB13,WKC15,CGHZ15,BCSY18} to facilitate data routing in network switches.
Despite the two types above, most works on batch codes focus on the \emph{primitive batch codes},
where primitive means that each server only stores one code symbol. Using the original notations from Ishai et al.,
a primitive batch code is indeed an $(n,N,k,N,1)$ batch code.
Primitive batch codes are of particular concerns, especially their trade-offs between the code length $N$ and the number of parallel requests $k$. Recently, many novel constructions of primitive batch codes with asymptotically rate $1$ are obtained through different methods, for examples, see \cite{VY16,AY19,PPV20,LW21,HPPVY21,KE23}. Meanwhile, by analyzing the dimension of a primitive batch code's dual code and its order-$O(k)$ tensor, \cite{RVW22} and \cite{LW21batch} provided lower bounds on the code length which are optimal for constant $k$. We refer the interested readers to Table II - Table V in \cite{HPPVY21}, which provide the best known trade-off between $N$ and $k$ for primitive batch codes.

Some further variants of primitive batch codes are as follows, depending on the restrictions on the batch of requests.
The \emph{private information retrieval (PIR) codes} introduced in \cite{FVY15} (also known as codes with $k$-DRGP, see \cite{LW21,RVW22}) only requires that every information symbol has $k$ mutually disjoint recovery sets,
which can be satisfied by taking $i_1=i_2=\dots=i_k$ as the $k$ parallel requests under the batch code setting. Bounds and constructions for PIR codes have been explored in several papers, including \cite{FVY15,LR17,VRK17,Skachek18,AY19,KY21,HPPVY21,RST22}. By allowing each request to be an arbitrary linear combination of the information symbols, \emph{functional PIR codes} and \emph{functional batch codes} were introduced
in two recent works \cite{ZEY20} and \cite{YY21}, where the former only deals with $k$ identical requests and the latter deals with an arbitrary multi-set of $k$ requests, in the same way as PIR codes and batch codes.

Following the primitive case, the research returns to the array codes versions, and to be more specific, uniform array versions, for PIR codes \cite{FVY15} and functional PIR/batch codes \cite{NY22}. In these array codes versions, each node stores a bucket containing multiple code symbols. Whenever the node is in a recovery set for a particular request, it probes all or at most a certain amount of stored symbols from its bucket, but responds with \emph{only one} linear combination of these symbols.
The constraint on the size of the response is for the the sake of reducing the communication bandwidth. One may argue that the original definition of $(n,N,k,m,1)$-batch codes by Ishai et al. \cite{IKOS04} is indeed also an array code version, or to be more specific, a non-uniform array version since the buckets are allowed to be of distinct sizes. However, there is a fundamental difference in the regime of how each node responds. The setting of the original batch codes in \cite{IKOS04} requires that each node only responds an exact copy of one code symbol in its bucket. Then each user collects all the responded symbols from the nodes in a particular recovery set, and does the necessary computation at the user's end to derive his request. This model is reasonable when nodes are storage drives such as disks and flash memories. On the contrary, the PIR array codes \cite{FVY15} and the array versions of functional PIR codes and functional batch codes \cite{NY22} are more compatible to large scale distributed storage systems, where a storage node consists of not only storage drives but also CPUs, RAID controllers and other hardware which are capable of performing computations. In such a setting, the symbol responded by each node could be a linear combination of all or a certain amount of the symbols in its bucket, computed locally by the node itself. In this paper, we adhere to the latter setting and consider the array versions of batch codes, i.e., \emph{batch array codes} (BACs for short). The next example illustrates the difference between the original batch codes of Ishai et al. \cite{IKOS04} and the BACs considered in this paper.

\begin{example}\label{ex1.1}
Let $n=4, k=4$ and $m=5$. The following construction gives a $(4, 14, 4, 5, 1)$-batch code $\cC_1$:
    \begin{table}[h]
        \centering
        \begin{tabular}{|c|c|c|c|c|}
        \hline
        $\bbc_1$&$ \bbc_2$&$\bbc_3$&$\bbc_4$&$\bbc_5$\\
        \hline
        $x_1 $& $x_1$ & $x_1$ & $x_2$ & $x_1+x_4$ \\\hline
        $x_2 $& $x_2$ & $x_3$ & $x_3$ & $x_2+x_3$ \\\hline
        $x_3 $& $x_4$ & $x_4$ & $x_4$ &  \\
        \hline
        \end{tabular}
    \end{table}\\

As shown by Ishai et al. in \cite{IKOS04}, for $(n, N, k, k+1,1)$-batch codes, we have $N\geq (k-\frac{1}{2})n$. Hence, $\cC_1$ is optimal in terms of the total code length.

Now, consider a code $\cC_2$ given by the following construction:
    \begin{table}[h]
        \centering
        \begin{tabular}{|c|c|c|c|c|}
        \hline
        $\bbc_1$&$ \bbc_2$&$\bbc_3$&$\bbc_4$&$\bbc_5$\\
        \hline
        $x_1 $& $x_1$ & $x_1$ & $x_2$ & $x_1+x_2+x_3+x_4$ \\\hline
        $x_2 $& $x_2$ & $x_3$ & $x_3$ &  \\\hline
        $x_3 $& $x_4$ & $x_4$ & $x_4$ &  \\
        \hline
        \end{tabular}
    \end{table}\\

As is shown later in Section \ref{subsec: constructions1}, $\cC_2$ is an $(4, 13, 4, 5)$ non-uniform BAC. Clearly, its total code length is one less than that of $\cC_1$. To see the difference of these two codes, consider the recovery process for the requests $\ms{x_1,x_1,x_1,x_1}$. For $\cC_1$, these requests can be supported by reading $x_1$ stored in $\bbc_1,\bbc_2,\bbc_3$, and retrieving $x_1$ by reading $x_4$ from $\bbc_4$ and $x_1+x_4$ from $\bbc_5$. During this process, each node only sends back one of its stored symbols and no computations are performed on the nodes. For $\cC_2$,  these requests can be supported by reading $x_1$ from $\bbc_1,\bbc_2,\bbc_3$, and retrieving $x_1$ by reading $x_2+x_3+x_4$ from $\bbc_4$ and $x_1+x_2+x_3+x_4$ from $\bbc_5$. During this process, the $4$th node first performs local computations over its stored data and then sends back $x_2+x_3+x_4$. The ability to access more symbols and do local computations leads to a smaller code length of $\cC_2$ compared to $\cC_1$.
\end{example}

Despite the analysis of the array codes version for PIR codes \cite{FVY15} and functional PIR/batch codes \cite{NY22}, the array version for batch codes seems to be skipped in relevant literatures, with only a few known results (see \cite{IKOS04,RSDG16,NY22}). The motivation of this paper is thus to fill the gap by analyzing BACs.
To be more specific, we analyze BACs in which each storage node has no access limit (thus the role of $t$ is no longer necessary) and can perform local computations over the data it stored. We demonstrate that, from both the perspectives of information theoretic bounds and explicit code constructions, this setting enables BACs to support the same type of requests as the original batch codes but with smaller code lengths. This implies that BACs have smaller storage overhead compared to batch codes for the same scenarios. We will consider both uniform and non-uniform arrays. Specifically, we have the following results.

\begin{itemize}
    \item Firstly, using information-theoretic arguments, we establish several lower bounds on the code length of $(n,N,k,m)$-BACs, including a lower bound for general parameter regime and two improved lower bounds for the cases $k<m<2k$ and $m=k+2$. All these bounds are shown to be tight for certain parameter regimes by explicit constructions.
    \item Secondly, we provide three constructions of BACs. The first two constructions are explicit and are from the combinatorial perspective, and the third one is a random construction. Specially, the code obtained from the first construction is an optimal $(n,N,k,k+1)$-BAC for any $k|n$ and the code obtained from the second construction is an optimal $(n,N,k,k+2)$-BAC for the case when $k=3$ and $n=5$. Furthermore, through the random construction, we demonstrate the existence of binary $(n,n(1+o(1)),k,m)$-BACs for $m=\Theta(n^{\frac{1}{2}})$ and $k=O(n^{\frac{1}{6}})$.
\end{itemize}

The rest of the paper is organized as follows. In Section \ref{sec: notations}, we list the notations we shall use throughout the paper and present several preliminary results. In Section \ref{sec: lower bounds}, we prove several lower bounds on the length of $(n,N,k,m)$-BACs. In Section \ref{sec: constructions}, three constructions of $(n,N,k,m)$-BACs are presented. Finally, in Section \ref{sec: conclusion}, we conclude the paper by highlighting some unresolved issues.

\section{Notations and preliminaries}\label{sec: notations}

For integers $1\leq m\leq n$, let $[m,n]=\{m,m+1,\dots,n\}$, $[n]=[1,n]$ and $[0]=\varnothing$. A \emph{multiset} is a set in which elements may repeat. We use the notation $I=\ms{i_1,i_2,\ldots,i_k}$, with $i_j\in [n]$ for all $j\in [k]$, to denote a multiset over $[n]$. Let $q$ be a prime power and let $\Fq$ be the finite field of size $q$. We denote $\Fq^{n}$ as the $n$-dimensional vector space over $\Fq$ and denote $\Fq^{\geq 1}$ denote the set of all vectors over $\Fq$ of finite length. For $\bbv\in\Fq^{\geq 1}$, we use $|\bbv|$ to denote the length of $\bbv$. For a finite set $A$, denote $\binom{A}{a} = \{B\subseteq A:|B|=a\}$ as the family of subsets of $A$ with size $a$ and $\binom{A}{\leq a} = \{B\subseteq A:|B|\leq a\}$ as the family of subsets of $A$ with size at most $a$. For vectors $\bbx=(x_1,x_2,\ldots,x_m)$ and $\bby = (y_1,y_2,\ldots,y_n)$, we say $\bbx$ is a \emph{sub-vector} of $\bby$ if there are indices $1\leq i_1< \cdots < i_m \leq n$ such that $\bbx=(y_{i_1},\ldots,y_{i_m})$.

Let $\mathcal{C}$ be a code of length $n$ over $\Fq$. For a codeword $\bbx=(x_1,x_2,\ldots,x_n)\in \cC$ and a subset $R\subseteq [n]$, let $\bbx|_{R}$ be the vector obtained by projecting the coordinates of $\bbx$ onto $R$, and define $\mathcal{C}|_{R}=\{\mathbf{x}|_{R}: \mathbf{x}\in \mathcal{C}\}$. Moreover, for convenience we will use the notation $x_i\in \bbx|_R$ to indicate that $i\in R$.

Now we introduce the formal definitions of batch array codes and PIR array codes.

\begin{definition}\label{def_BAC}[Batch Array Codes (BAC)]
 An $(n,N,k,m)$-BAC over $\mathbb{F}_q$ is defined by a linear encoding function $\cC$: $\mathbb{F}_q^n\rightarrow (\mathbb{F}_q^{\geq 1})^m$ such that for any $\bbx=(x_1,x_2,\ldots,x_n)\in \mathbb{F}_q^{n}$ the following hold:
 \begin{itemize}
     \item $\cC$ encodes $\bbx$ as $\cC(\bbx)=(\bbc_1,\ldots,\bbc_m)$, where $\bbc_\ell\in \Fq^{N_\ell}$ is referred to as the $\ell$-th bucket, for each $\ell \in [m]$.
     The total length of all $m$ buckets is $\sum_{\ell=1}^{m}N_\ell=N$ (where the length of each bucket is independent of $\bbx$).
     \item For any multiset $I=\ms{i_1,i_2,\ldots,i_k}\subseteq [n]$ there is a partition of $[m]$ into $k$ subsets $R_1,R_2,\ldots,R_k$, and for each $j\in[k]$ there is a set of linear functions $\{f_{\ell}: \Fq^{N_\ell}\rightarrow \Fq\}_{\ell\in R_j}$, such that $x_{i_j}$ is an $\mathbb{F}_q$-linear combination of the symbols $\{f_{\ell}(\bbc_\ell):\ell\in R_j\}$. Additionally, we say $x_{i_j}$ can be recovered by $\{\bbc_\ell:\ell\in R_j\}$ and $R_{j}$ is called the recovery set of $x_{i_j}$.
 \end{itemize}

Furthermore, if all $N_\ell$'s are equal then $\cC$ is called a uniform BAC, which coincides with the definition of uniform batch array codes in \cite{NY22} (see Definition 1-b \cite{NY22} for details). Otherwise, it is called a non-uniform BAC. In the sequel by BACs we mean the non-uniform version unless otherwise stated.
\end{definition}

\begin{definition}\label{def_CPIR}[PIR array codes]
 An $(n,N,k,m)$-PIR array code over $\Fq$ is defined by a linear encoding function $\cC$: $\Fq^n\rightarrow (\Fq^{\geq 1})^m$ such that for any $\bbx=(x_1,x_2,\ldots,x_n)\in \Fq^{n}$ the following hold:
 \begin{itemize}
     \item $\cC$ encodes $\bbx$ as $\cC(\bbx)=(\bbc_1,\ldots,\bbc_m)$, where $\bbc_\ell\in \Fq^{N_\ell}$ is referred to as the $\ell$-th bucket, for each $\ell \in [m]$.
     The total length of all $m$ buckets is $\sum_{\ell=1}^{m}N_i=N$ (where the length of each bucket is independent of $\bbx$).
     \item For each $i\in [n]$ there is a partition of $[m]$ into $k$ subsets $R_1,R_2,\ldots,R_k$, and for each $j\in[k]$ there is a set of linear functions $\{f_{\ell}: \Fq^{N_\ell}\rightarrow \Fq\}_{\ell\in R_j}$, such that $x_i$ is an $\Fq$-linear combination of the symbols $\{f_{\ell}(\bbc_\ell):\ell\in R_j\}$. Additionally, we say $x_i$ can be recovered by $\{\bbc_\ell:\ell\in R_j\}$ and $R_1,R_2,\ldots,R_k$ are called the recovery sets of $x_i$.
 \end{itemize}

Furthermore, if all $N_\ell$'s are equal then $\cC$ is called a uniform PIR array code, which coincides with the definition of PIR array codes in \cite{FVY15} (see also \cite{Blackburn19,Zhang19,CKYZ19,NY22,Wang23}). Otherwise, it is called a non-uniform PIR array code. In the sequel by PIR array codes we mean the non-uniform version unless otherwise stated.
\end{definition}


\begin{definition}
    For positive integers $n, k$ and $m$ with $k\leq m$, we define
    \begin{align*}
        N_{P}(n, k, m) = \min\{N:\text{There is an $(n, N, k, m)$-PIR array code.}\}
    \end{align*}
    as the minimum code length of $(n, N, k, m)$-PIR array codes and
    \begin{align*}
        N_{B}(n, k, m) = \min\{N:\text{There is an $(n, N, k, m)$-BAC.}\}
    \end{align*}
    as the minimum code length of $(n, N, k, m)$-BACs. Clearly, we have $N_{B}(n, k, m)\geq N_{P}(n, k, m)$. For convenience, if $n$, $k$, and $m$ are clear from the context or not relevant, we omit them from the notations and write $N_p$ and $N_B$ instead.
\end{definition}

Similarly as the other related codes, given $n$, $k$, and $m$, the general problem in BACs and PIR array codes is to find the minimum code length $N_{B}$ or $N_{P}$, and thus maximize the code rate $\rho\triangleq\frac{n}{N_B}$ or $\rho\triangleq\frac{n}{N_P}$. Next, we show several properties of an $(n, N, k, m)$-PIR array code, which are crucial for the proof of our lower bounds for $N_{P}$ and $N_{B}$ in Section \ref{sec: lower bounds}.

\begin{lemma}\label{lem: size of $m-k+1$ buckets}
    Let $\cC$ be an $(n, N, k, m)$-PIR array code over $\Fq$. Then, for any $\bbx=(x_1,\ldots,x_n)\in \Fq^n$ and any subset $R\subseteq [m]$ of size $m-k+1$, $R$ contains a recovery set of $x_i$ for every $i\in [n]$.
\end{lemma}
\begin{proof}
    The statement directly follows from the definition of the PIR array code when $k=1$. For the case when $k\geq 2$, we assume that there exists some $i_0\in [n]$ such that $R$ does not contain any recovery set for $x_{i_0}$. Next, we proceed to prove the lemma by deriving a contradiction.

    W.l.o.g. we can assume that $R=[m-k+1]$ and $i_{0}=1$. By the definition of PIR array codes, there are $k$ recovery sets $R_1, R_2, \ldots, R_k$ for $x_1$ and $R_1,\ldots, R_k$ forms a partition of $[m]$. Since $R$ doesn't contain any recovery set of $x_{1}$, this implies that $R_j\cap [m-k+2,m]\neq \varnothing$ for every $j\in [k]$. However, there are only $k-1$ elements in $[m-k+2,m]$ and it's impossible to have $k$ mutually disjoint sets $R_1,\ldots, R_k$ such that $R_j\cap [m-k+2,m]\neq \varnothing$ for every $j\in [k]$. This contradicts the assumption at the beginning and hence proves the lemma.
\end{proof}

Let $\cC$ be an $(n,N,k,m)$-BAC or PIR array code over $\Fq$. Given a batch of request, there might be multiple choices of recovery sets and it is possible that certain nodes do not participate in any recovery sets. In this paper, whenever we consider a BAC or a PIR array code we assume that there is a corresponding \emph{fixed} way to choose the recovery sets for any batch of requests, and w.l.o.g. we assume that every node participates in a unique recovery set, even if it has no contributions. These settings lead to a rigourous characterization of the responding regime of an $(n,N,k,m)$-PIR array code as follows. The case of $(n,N,k,m)$-BACs can be characterized similarly.

For $\bbx\in \Fq^n$, let $\cC(\bbx)=(\bbc_1,\ldots,\bbc_m)$ be the codeword in $\cC$ encoded from $\bbx$. By the linearity of $\cC$, every codeword symbol in the bucket $\bbc_\ell$ ($\ell\in [m]$) is a linear combination of $x_1,\ldots,x_n$. Thus, for each $\ell\in [m]$, there is an $n\times N_\ell$ matrix $\bbG_\ell$ over $\Fq$ such that
$$\bbx\cdot \bbG_\ell=\bbc_\ell.$$
We call $\bbG_\ell$ the \emph{generator matrix} of the $\ell$-th bucket and let
$$\bbG=(\bbG_1~\bbG_2~\cdots~\bbG_m)$$
be the \emph{generator matrix} of $\cC$.

By Definition \ref{def_CPIR}, for each $\ell\in [m]$, $\bbc_\ell$ lies in a unique recovery set of $x_i$ for every $i\in [n]$. That is, for each $\ell\in [m]$, there is a set of linear functions $\{f_{\ell}^{(1)},f_{\ell}^{(2)},\ldots,f_{\ell}^{(n)}\}$ such that the symbol $f_{\ell}^{(i)}(\bbc_{\ell})$ is used in the recovery process of $x_i$. We call $f_{\ell}^{(i)}(\bbc_{\ell})$ the response from the $\ell$-th bucket in the recovery of $x_i$. By the linearity of the code, the $\ell$-th bucket can be associated with a column vector $\bbf_{\ell}^{(i)}\in \Fq^{N_\ell}$ such that $f_{\ell}^{(i)}(\bbc_{\ell})=\bbc_{\ell}\cdot \bbf_{\ell}^{(i)}$. Define $\bbF_\ell$ as the $N_\ell\times n$ matrix $(\bbf_{\ell}^{(1)}~\bbf_{\ell}^{(2)}~\cdots~\bbf_{\ell}^{(n)})$. The matrix $\bbF_\ell$ captures all the responses from the $\ell$-th bucket for the recovery of every $x_i$, thus we call it the \emph{response matrix} of the $\ell$-th bucket.

\begin{lemma}\label{lem: property of response matrix}
    Let $\cC$ be an $(n,N,k,m)$-PIR array code with $N=N_{p}(n,k,m)$. Then, for each $\ell\in [m]$, we have $N_\ell\leq n$ and $\mathrm{rank}(\bbF_\ell)=N_\ell$. Moreover, let $\bbf_{\ell}^{(i_1)},\bbf_{\ell}^{(i_2)},\ldots,\bbf_{\ell}^{(i_{N_\ell})}$ be $N_\ell$ linearly independent columns in $\bbF_\ell$. Then, through basis transformation, we can assume that $(\bbf_{\ell}^{(i_1)}~\bbf_{\ell}^{(i_2)}~\cdots~\bbf_{\ell}^{(i_{N_\ell})})=\bbI$.
\end{lemma}

\begin{proof}
Since each bucket can access all of its stored symbols and respond with any $\Fq$-linear combination of these symbols, what we essentially store in a bucket can be viewed as an $\Fq$-linear space spanned by these symbols. Based on this observation, we proceed with the proof.

First, we show that $N_\ell\leq n$ holds for all $\ell\in [m]$. Suppose otherwise, then w.l.o.g assume that $N_{1}>n$. Consider the new code $\mathcal{C}'$ obtained from $\mathcal{C}$ by replacing the bucket $\mathbf{c}_1$ with $\mathbf{x}$ and maintaining all the other buckets. For any request $x_i$, $f_{1}^{(i)}(\bbc_{1})$ is an $\Fq$-linear combination of all the information symbols $x_1,\dots,x_n$ and thus the new bucket storing $\bbx$ can respond with the same code symbol as the original bucket. This implies that $\cC'$ is an $(n,N',k,m)$-PIR array code with $N'=n+\sum_{i=2}^{m}N_i=N-N_1+n<N_p(n,k,m)$, which leads to a contradiction.

Next, we show that $\mathrm{rank}(\bbF_{\ell})=N_{\ell}$ holds for each $\ell\in [m]$. Suppose the response matrix of the $\ell$-th bucket $\bbc_{\ell}$ has $\mathrm{rank}(\bbF_\ell)<N_\ell$. Then, we can replace the stored symbols in the $\ell$-th bucket by $s=\mathrm{rank}(\bbF_{\ell})$ symbols $f_{\ell}^{(i_1)}(\bbc_{\ell}),\ldots,f_{\ell}^{(i_{s})}(\bbc_{\ell})$ such that $\bbf_{\ell}^{(i_1)},\ldots,\bbf_{\ell}^{(i_{s})}$ form a basis of the column space of $\bbF_{\ell}$. After this replacement, the $\ell$-th bucket stores fewer symbols but can still respond with the same $f_{\ell}^{(i)}(\bbc_{\ell})$ as the original $\ell$-th bucket for any request $x_i$. This leads to a new PIR array code of shorter length with the same other parameters, which contradicts the assumption $N=N_{p}(n,k,m)$. Thus, $\mathrm{rank}(\bbF_\ell)=N_\ell$.

Let $\bbf_{\ell}^{(i_1)},\bbf_{\ell}^{(i_2)},\ldots,\bbf_{\ell}^{(i_{N_\ell})}$ be $N_\ell$ linearly independent columns in $\bbF_\ell$. Then, essentially the bucket stores all the $\Fq$-linear combinations of $f_{\ell}^{(i_1)}(\bbc_{\ell}),\dots,f_{\ell}^{(i_{N_\ell})}(\bbc_{\ell})$. Replace $\bbc_{\ell}$ as $\bbc'_{\ell}=\{f_{\ell}^{(i_1)}(\bbc_{\ell}),\dots,f_{\ell}^{(i_{N_\ell})}(\bbc_{\ell})\}$. Then through basis transformation, in the new response matrix $\bbF'_\ell$
we have $(\bbf_{\ell}^{(i_1)}~\bbf_{\ell}^{(i_2)}~\cdots~\bbf_{\ell}^{(i_{N_\ell})})=\bbI$.
\end{proof}

The following example illustrates the lemma above.

\begin{example}
Suppose we have a binary $(n=4,N,k,m)$-PIR array code and the first bucket originally stores $\bbc_1=(x_1+x_2,x_2+x_3,x_3+x_4)$. Assume that its responses in the recovery for $x_1,x_2,x_3,x_4$ are $x_1+x_2,x_2+x_4,x_1+x_3,x_1+x_2+x_3+x_4$ accordingly. Then the response matrix $\bbF_1=\left(\begin{array}{cccc}
1 & 0 & 1 & 1 \\
0 & 1 & 1 & 0 \\
0 & 1 & 0 & 1 \\
\end{array}
\right)
$. Since the first three columns are linearly independent, we can replace the contents in the first bucket as $\bbc'_1=(x_1+x_2,x_2+x_4,x_1+x_3)$. Then through basis transformation we have the new response matrix $\bbF'_1=\left(\begin{array}{cccc}
1 & 0 & 0 & 0 \\
0 & 1 & 0 & 1 \\
0 & 0 & 1 & 1 \\
\end{array}
\right)$.
\end{example}

\section{Lower bounds on the code length of BACs and PIR array codes}\label{sec: lower bounds}

In this section, we study the minimum code length of an $(n,N,k,m)$-PIR array code. Using information theoretic arguments, we prove three lower bounds of $N_P(n,k,m)$, which are also the lower bounds of $N_B(n,k,m)$ since $N_{B}(n, k, m)\geq N_{P}(n, k, m)$. Throughout the section, we use $H(\cdot)$ to denote the $q$-ary entropy function.

\subsection{A general lower bound}

\begin{theorem}\label{Thm: general LB_1}
    For positive integers $n$, $m$, and $k$ with $m\geq k$, we have $$N_P(n,k,m)\geq mn/(m-k+1).$$
\end{theorem}
\begin{proof}
    Let $\cC$ be an $(n, N, k, m)$-PIR array code. Recall that $\bbx=(x_1,\ldots,x_n)\in \Fq^{n}$ is encoded as $\cC(\bbx)=(\bbc_1,\ldots,\bbc_m)$ in $\cC$ and $|\bbc_\ell|=N_\ell$ for each $\ell\in [m]$. W.l.o.g., we assume that $N_1\geq N_2\geq\ldots\geq N_m$. Let $\bbx$ be a random vector uniformly distributed over $\Fq^n$. By Lemma \ref{lem: size of $m-k+1$ buckets}, for any $i\in [n]$, $[k,m]$ contains a recovery set of $x_i$. Thus, we have
$$ n = H(\bbc_k,\bbc_{k+1},\ldots,\bbc_{m}) \leq H(\bbc_k)+H(\bbc_{k+1})+\cdots+H(\bbc_{m}) \leq \sum_{\ell=k}^{m}N_k. $$

    By the monotone decreasing property of $N_\ell$, it holds that $N_{\ell}\geq N_{k}\geq n/(m-k+1)$ for each $\ell\in [k-1]$. Therefore, we have
    $$N = \sum_{\ell=1}^{m}N_\ell = \sum_{\ell=1}^{k-1}N_\ell + \sum_{\ell=k}^m N_\ell \geq \frac{(k-1)n}{m-k+1}+n = \frac{mn}{m-k+1}.$$

    This concludes the proof.
\end{proof}

As an immediate consequence of Theorem \ref{Thm: general LB_1}, we have the following corollary.

\begin{corollary}\label{Coro: BAC_general LB_1}
    For positive integers $n$, $m$, and $k$ with $m\geq k$, we have $N_B(n,k,m)\geq mn/(m-k+1)$.
\end{corollary}

\begin{remark}\label{rmk3-1}
    The following constructions show that the lower bound in Corollary \ref{Coro: BAC_general LB_1} is tight for the following cases.
    \begin{itemize}
        \item When $m=k$, each server itself must be a recovery set for each information symbol. Thus we can set each bucket as $\{x_1,\dots,x_n\}$, which leads to $N_P(n,k,k)=N_B(n,k,k)=kn$.
        \item When $k=1$, arbitrarily distribute the information symbols $x_1,\dots,x_n$ into $m$ buckets such that each symbol appears exactly once. Clearly this is an $(n, N=n, 1, m)$-BAC.
        \item When $k=2$ and $n=m-1$, for any $\bbx\in\mathbb{F}_q^{m-1}$, define $\cC_2(\bbx) = (\bbc_1, \ldots, \bbc_m)$ by taking $\bbc_i = x_i$ for every $i\in [m-1]$ and $\bbc_{m} = x_1+x_2+\ldots+x_{m-1}$. Then, for any $i\in [m-1]$, $x_i$ can be recovered simultaneously by $\bbc_i$ and $\{\bbc_{\ell}\}_{\ell\in [m]\setminus \{i\}}$. Thus, for any multi-set $\ms{i_1,i_2}\subseteq [n]$, $x_{i_1}$ has recovery set $\{i_1\}$ and $x_{i_2}$ has recovery set $[m]\setminus \{i_1\}$. This shows that $\cC_2$ is an $(m-1, m, 2, m)$-BAC.
    \end{itemize}
    Moreover, by the gadget Lemma \ref{lem: gadget lemma} in Section \ref{subsec: constructions1}, we can use $\cC_2$ to obtain $(n, \frac{mn}{m-1}, 2, m)$-BACs for a general $n$ such that $m-1 \mid n$.
\end{remark}

\subsection{An improved bound for $k<m<2k$}

\begin{theorem}\label{Thm: general LB_2}
    For positive integers $n$, $m$, and $k$ with $k<m<2k$, we have $$N_P(n,k,m)\geq (2k-m + \frac{1}{\binom{m-1}{2k-m}})n.$$
\end{theorem}
\begin{proof}
    Let $\cC$ be an $(n, N, k, m)$-PIR array code with $k < m < 2k$. For $\bbx=(x_1,\ldots,x_n)\in \Fq^n$, denote $\cC(\bbx)=(\bbc_1,\ldots,\bbc_m)$ as the codeword in $\cC$ encoded from $\bbx$.

     First we claim that for each $i\in [n]$, $x_i$ has at least $2k-m$ recovery sets of size $1$. Otherwise, assume that $x_{i_0}$ has less than $2k-m$ recovery sets of size $1$ for some $i_0\in [n]$. Then, the union of all recovery sets of $x_{i_0}$ has size at least
    $$2k-m-1 + 2(k-(2k-m-1)) = m+1>m.$$
    This contradicts the fact that all the $k$ recovery sets of $x_{i_0}$ form a partition of $[m]$.

    For each $i\in [n]$, denote $U_i$ as the union of all the size $1$ recovery sets of $x_i$. Clearly, $U_i\subseteq [m]$ and $|U_i|\geq 2k-m$.
    Denote the smallest $2k-m$ coordinates of $U_i$ as $U_i[1:2k-m]$. For every set $A\in \binom{[m]}{2k-m}$, let $\mathrm{Ind}(A)\subseteq [n]$ be the set of all indices $i\in[n]$ such that $A=U_i[1:2k-m]$. That is, if $i\in \mathrm{Ind}(A)$, then the smallest $2k-m$ coordinates of the size 1 recovery sets for the symbol $x_i$ is exactly $A$. Since $|U_i|\geq 2k-m$ for each $i\in [n]$, then there exists a unique $A\in \binom{[m]}{2k-m}$ such that $i\in \mathrm{Ind}(A)$. This leads to
    \begin{equation}\label{eq1_thmIII.4}
        n=\sum_{A\in \binom{[m]}{2k-m}} |\mathrm{Ind}(A)|.
    \end{equation}

    Moreover,
    for each $\ell\in [m]$, define $\bbw_\ell$ as the restriction of $\bbx=(x_1,\ldots,x_n)$ onto the coordinates
    $$\bigcup_{A\in \binom{[m]}{2k-m}, ~ \ell\in A} \mathrm{Ind}(A).$$
    In other words, for every $x_i\in \bbw_{\ell}$, there exists a set $A\in \binom{[m]}{2k-m}$ containing $\ell$ such that $i\in \mathrm{Ind}(A)$. Thus, we have $\ell\in U_{i}$, which implies that $x_i$ can be recovered by the bucket $\bbc_{\ell}$. Moreover, since $\mathrm{Ind}(A)$'s are mutually disjoint, we also have
    \begin{align}\label{eq2_thmIII.4}
        |\bbw_\ell|=
        \sum_{A\in \binom{[m]}{2k-m}, ~ \ell\in A} |\mathrm{Ind}(A)|.
    \end{align}


    Next, similar to the proof of Theorem \ref{Thm: general LB_1}, we employ the information entropy to prove the lower bound.
    Let $\bbx$ be uniformly distributed over $\mathbb{F}_q^n$. For any fixed $A\in \binom{[m]}{2k-m}$, let $\bbx_A \triangleq \bbx|_{\mathrm{Ind(A)}}$, $\bby_A\triangleq \bbx|_{[n]\setminus \mathrm{Ind(A)}}$, and let $\bbz_A$ be the concatenation of all $\bbc_\ell$'s such that $\ell\in [m]\setminus A$. Since $\bbx_A$ and $\bby_A$ are independent, we have $|\mathrm{Ind}(A)|=H(\bbx_A)=H(\bbx_A|\bby_A)$. Note that for each $i\in \mathrm{Ind}(A)$, $x_i$ has $k$ recovery sets which form a partition of $[m]$. Since $2k-m<k$, the buckets indexed by $[m]\setminus A$ must contain at least one recovery set for each $x_i$ with $i\in \mathrm{Ind}(A)$. By the definition of $\bbz_A$, we have
    \begin{align}\label{eq3_thmIII.4}
    H(\bbx_A)=H(\bbx_A|\bby_A)\leq  H(\bbz_A| \bby_A) \leq \sum_{\ell\in [m]\setminus A}H(\bbc_\ell| \bby_A).
    \end{align}

    For any $\ell\in [m]\setminus A$, by the definition of $\bbw_\ell$, the information symbols in $\bbx_A$ are disjoint with those in $\bbw_\ell$. In other words, $\bbw_\ell$ is a sub-vector of $\bby_{A}$. This implies that $H(\bbc_\ell| \bby_A)\leq H(\bbc_\ell| \bbw_\ell)$ holds for every $\ell\in [m]\setminus A$. Therefore, by (\ref{eq3_thmIII.4}), we have
    \begin{align}\label{eq4_thmIII.4}
        H(\bbx_A) \leq \sum_{\ell\in [m]\setminus A}H(\bbc_\ell| \bby_A)\leq \sum_{\ell\in [m]\setminus A}H(\bbc_\ell| \bbw_\ell).
    \end{align}


    Next, by summing both sides of (\ref{eq4_thmIII.4}) over all $A\in \binom{[m]}{2k-m}$ and using (\ref{eq1_thmIII.4}), we obtain
    \begin{align}\label{eq5_thmIII.4}
    n\leq &\sum_{A\in \binom{[m]}{2k-m}}\sum_{\ell\in [m]\setminus A}H(\bbc_\ell| \bbw_\ell)\nonumber\\
    = &\sum_{\ell\in [m]} \sum_{A\in \binom{[m]}{2k-m}, ~ \ell\notin A}H(\bbc_\ell| \bbw_\ell)=\binom{m-1}{2k-m}\sum_{\ell\in [m]}H(\bbc_\ell|\bbw_\ell).
    \end{align}

    On the other hand, by the property of conditional entropy, we have
    \begin{align*}
    H(\bbc_\ell)
    =  &H(\bbc_\ell| \bbw_\ell) + H(\bbw_\ell) - H(\bbw_\ell| \bbc_\ell)\\
    =  &H(\bbc_\ell| \bbw_\ell) + H(\bbw_\ell),
    \end{align*}
    where the last equality holds since any information symbol $x_i\in\bbw_\ell$ can be recovered from the the bucket $\bbc_\ell$ and thus $H(\bbw_\ell| \bbc_\ell)=0$.

    Then, by (\ref{eq2_thmIII.4}), we have
    \begin{equation}\label{eq6_thmIII.4}
    H(\bbc_\ell|\bbw_\ell)=H(\bbc_\ell)-\sum_{A\in \binom{[m]}{2k-m}, ~\ell\in A}|\mathrm{Ind}(A)|.
    \end{equation}
    Plugging (\ref{eq6_thmIII.4}) into (\ref{eq5_thmIII.4}), we have
    \begin{align*}
        n\leq &\binom{m-1}{2k-m}\sum_{\ell\in [m]}\left(H(\bbc_\ell)-\sum_{A\in \binom{[m]}{2k-m}, ~\ell\in A}|\mathrm{Ind}(A)|\right).
    \end{align*}
    Since $N=\sum_{\ell\in [m]} H(\bbc_\ell)$, we can finally derive that
    \begin{align*}
        N\geq  &\sum_{\ell\in [m]}\sum_{A\in \binom{[m]}{2k-m}, ~\ell\in A}|\mathrm{Ind}(A)|+\frac{n}{\binom{m-1}{2k-m}}\\
        = & (2k-m)\sum_{A\in \binom{[m]}{2k-m}}|\mathrm{Ind}(A)|+\frac{n}{\binom{m-1}{2k-m}}\\
        \overset{(a)}{=} & (2k-m + \frac{1}{\binom{m-1}{2k-m}})n,
    \end{align*}
    where $(a)$ follows from (\ref{eq1_thmIII.4}). This concludes the proof.
\end{proof}

\begin{corollary}\label{Coro: general LB_2}
    For any positive integers $m$ and $k$ with $m=k+1$, we have
    $$N_B(n, k, k + 1) \geq N_P(n, k, k + 1) \geq (k-1+1/k)n.$$
\end{corollary}

\begin{remark}\label{rmk3-2}
    In Section \ref{sec: constructions}, we will provide an explicit construction of BACs, which shows that the lower bounds in Theorem \ref{Thm: general LB_2} and Corollary \ref{Coro: general LB_2} are tight for the case when $m=k+1$.
\end{remark}

\begin{remark}
    In Theorem A.7 of \cite{IKOS04}, Ishai et al. proved that for an $(n, N, k, k+1, 1)$-batch code it holds that $N\geq (k-\frac{1}{2})n$.
    This lower bound is tight, e.g., see Example \ref{ex1.1} in the introduction. However, as mentioned in Remark \ref{rmk3-2}, the lower bound established in Corollary \ref{Coro: general LB_2} is tight for BACs. This suggests that BACs could be strictly preferable than the original batch codes in view of the storage overhead.
\end{remark}

\subsection{An improved bound for $m=k+2$}

Recall that by Remark \ref{rmk3-1} and Remark \ref{rmk3-2}, we have $N_P(n,k,k)=kn$ and $N_P(n,k,k+1)=(k-1+\frac{1}{k})n$. In this subsection we further improve the lower bound of $N_P(n,k,k+2)$ as follows.

\begin{theorem}\label{Thm: LB of $(n,N,k,k+2)$-PIRA}
     For positive integers $n$ and $k\geq 3$, we have
     $$N_P(n,k,k+2)\geq (k-2+\frac{4k+16}{3k^2+k+4})n.$$
\end{theorem}

\begin{remark}\label{rmk3-3}
When $m=k+2$, the lower bound obtained in Theorem \ref{Thm: general LB_2} shows that $N_P(n,k,k+2)\geq (k-2 + \frac{1}{\binom{k+1}{3}})n.$
Since $\frac{4k+16}{3k^2+k+4}\geq \frac{1}{\binom{k+1}{3}}$ when $k\geq 3$, Theorem \ref{Thm: LB of $(n,N,k,k+2)$-PIRA} improves the result of Theorem \ref{Thm: general LB_2}. Moreover, when $k=3$ and $n=5$, the lower bound in Theorem \ref{Thm: LB of $(n,N,k,k+2)$-PIRA} is shown to be optimal by an explicit construction, see Example \ref{ex: $(n=5, N=10, k=3, m=5)$-computational batch code} in Section \ref{subsec: constructions2}.
\end{remark}

Before delving into the detailed proof, we first demonstrate a few properties of $(n,N,k,k+2)$-PIR array codes and sketch the idea of the proof of Theorem \ref{Thm: LB of $(n,N,k,k+2)$-PIRA}.

Let $\cC$ be an $(n,N,k,k+2)$-PIR array code. For $\bbx=(x_1,\ldots,x_n)\in \Fq^n$, denote $\cC(\bbx)=(\bbc_1,\ldots,\bbc_{k+2})$. By Definition \ref{def_CPIR}, each $x_i\in \bbx$ has $k$ recovery sets $R_1,R_2,\ldots, R_k$, which form a partition of $[k+2]$. W.l.o.g., we assume that for each $i\in [n]$, $x_i$ is assigned with a unique group of recovery sets $\{R_1^{(i)},R_2^{(i)},\ldots,R_k^{(i)}\}$ and $|R_1^{(i)}|\leq |R_2^{(i)}|\leq \cdots\leq |R_k^{(i)}|$. Then, for each $x_i$, the sizes of its recovery sets satisfy
\begin{equation}\label{eq0_(n,N,k,k+2)}
\begin{cases}
|R_1^{(i)}|=\cdots=|R_{k-1}^{(i)}|=1;\\
|R_k^{(i)}|=3,
\end{cases}~\mathrm{or}~~
\begin{cases}
|R_1^{(i)}|=\cdots=|R_{k-2}^{(i)}|=1;\\
|R_{k-1}^{(i)}|=|R_k^{(i)}|=2.
\end{cases}
\end{equation}

For $A\in \binom{[m]}{\leq 3}$, denote $\bbx_{A}$ as the projection of $\bbx$ onto the coordinates $\{i:A~\text{is a recovery set of}~x_i\}$.
For $A,B\in \binom{[m]}{\leq 3}$ such that $A\cap B = \varnothing$, denote $\bbx_{A, B}$ as the projection of $\bbx$ onto $\{i: \text{Both}~A~\text{and}~B~\text{are recovery sets of}~x_i\}$. Note that in the notation $\bbx_{A, B}$, the order of $A$ and $B$ does not matter. Clearly, $\bbx_{A, B}$ is a sub-vector of $\bbx_{A}$ and $\bbx_{B}$.
Also note that given two distinct elements $a\in[m]$ and $b\in[m]$, the notations $\bbx_{\{a\}, \{b\}}$ and $\bbx_{\{a,b\}}$ are different, where the former is the projection of $\bbx$ onto $\{i: \text{Both}~\{a\}~\text{and}~\{b\}~\text{are recovery sets of}~x_i\}$ and the latter is the projection of $\bbx$ onto $\{i:\{a,b\}~\text{is a recovery set of}~x_i\}$. By (\ref{eq0_(n,N,k,k+2)}), we have the following lemma.

\begin{lemma}\label{lem: LB for m=k+2}
For each $i\in [n]$, exactly one of the following holds: \begin{itemize}
        \item [1.] There is a unique $A\in \binom{[m]}{3}$ such that $x_i\in \bbx_{A}$ and $x_i\in \bbx_{\{\ell\}}$ for every $\ell\in [m]\setminus A$.
        \item [2.] There is a unique pair $\{B,C\}\subseteq \binom{[m]}{2}$ satisfying $B\cap C=\varnothing$ such that $x_i\in \bbx_{B,C}$ and $x_i\in \bbx_{\{\ell\}}$ for every $\ell\in [m]\setminus B\cup C$.
    \end{itemize}
\end{lemma}

Lemma \ref{lem: LB for m=k+2} classifies all information symbols in $\bbx$ based on the structure of their recovery sets. This enables us to derive the following identities.

\begin{corollary}\label{coro: LB for m=k+2}
Let $M_1=\sum_{\{a,b,c,d\}\in\binom{[m]}{4}}(|\bbx_{\{a,b\},\{c,d\}}|+|\bbx_{\{a,c\},\{b,d\}}|+|\bbx_{\{a,d\},\{b,c\}}|)$ and $M_2=\sum_{A\in\binom{[m]}{3}}|\bbx_{A}|$. Then, the following holds:
\begin{equation}\label{eq00b_(n,N,k,k+2)}
    n=|\bbx|=M_1+M_2.
\end{equation}
\begin{align}\label{eq00c_(n,N,k,k+2)}
    \sum_{\ell\in [m]}|\bbx_{\{\ell\}}|=(m-4)n+M_2.
\end{align}
\begin{equation}\label{eq00d_(n,N,k,k+2)}
    \sum_{\{a,b\}\in \binom{[m]}{2}}|\bbx_{\{a\},\{b\}}|=\binom{m-4}{2} M_1+\binom{m-3}{2} M_2.
\end{equation}
\begin{equation}\label{eq00e_(n,N,k,k+2)}
    \sum_{\{a,b,c\}\in \binom{[m]}{3}}|\bbx_{\{a\},\{b\},\{c\}}|=\binom{m-4}{3} M_1+\binom{m-3}{3} M_2.
\end{equation}
\begin{equation}\label{eq00f_(n,N,k,k+2)}
    \sum_{\{a,b,c\}\in \binom{[m]}{3}}(|\bbx_{\{a,b\},\{c\}}|+|\bbx_{\{a,c\},\{b\}}|+|\bbx_{\{b,c\},\{a\}}|)=2(m-4) M_1.
\end{equation}
\end{corollary}

\begin{proof}
By the classification of all information symbols in $\bbx$ from Lemma \ref{lem: LB for m=k+2}, (\ref{eq00b_(n,N,k,k+2)}) immediately follows. Moreover,
we have that \begin{equation}\label{eq00a0_(n,N,k,k+2)}
    |\bbx_{\{\ell\}}|=\sum_{\{a,b,c,d\}\in\binom{[m]\setminus \{\ell\}}{4}}\left(|\bbx_{\{a,b\},\{c,d\}}|+|\bbx_{\{a,c\},\{b,d\}}|+|\bbx_{\{a,d\},\{b,c\}}|\right)+\sum_{A\in\binom{[m]\setminus\{\ell\}}{3}}|\bbx_{A}|
\end{equation}
holds for each $\ell\in [m]$, and
\begin{equation}\label{eq00a1_(n,N,k,k+2)}
    |\bbx_{\{a\},\{b\}}|=\sum_{\{a',b',c',d'\}\in\binom{[m]\setminus \{a,b\}}{4}}\left(|\bbx_{\{a',b'\},\{c',d'\}}|+|\bbx_{\{a',c'\},\{b',d'\}}|+|\bbx_{\{a',d'\},\{b',c'\}}|\right)+\sum_{A\in\binom{[m]\setminus\{a,b\}}{3}}|\bbx_{A}|
\end{equation}
holds for each $\{a,b\}\in \binom{[m]}{2}$, and the following two equations
\begin{equation}
|\bbx_{\{a\},\{b\},\{c\}}|=\sum_{\{a',b',c',d'\}\in\binom{[m]\setminus \{a,b,c\}}{4}}\left(|\bbx_{\{a',b'\},\{c',d'\}}|+|\bbx_{\{a',c'\},\{b',d'\}}|+|\bbx_{\{a',d'\},\{b',c'\}}|\right)+\sum_{A\in\binom{[m]\setminus\{a,b,c\}}{3}}|\bbx_{A}| \label{eq00a2_(n,N,k,k+2)}
\end{equation}
\begin{equation}
|\bbx_{\{a,b\},\{c\}}|+|\bbx_{\{a,c\},\{b\}}|+|\bbx_{\{b,c\},\{a\}}|=\sum_{\{a',b'\}\in\binom{[m]\setminus \{a,b,c\}}{2}}\left(|\bbx_{\{a,b\},\{a',b'\}}|+|\bbx_{\{a,c\},\{a',b'\}}|+|\bbx_{\{b,c\},\{a',b'\}}|\right) \label{eq00a3_(n,N,k,k+2)}
\end{equation}
hold for each $\{a,b,c\}\in \binom{[m]}{3}$.

By summing both sides of (\ref{eq00a0_(n,N,k,k+2)}) over all $\ell\in [m]$, we have
\begin{align*}
\sum_{\ell\in [m]}|\bbx_{\{\ell\}}|=&\sum_{\ell\in [m]}\sum_{\{a,b,c,d\}\in\binom{[m]\setminus \{\ell\}}{4}}\left(|\bbx_{\{a,b\},\{c,d\}}|+|\bbx_{\{a,c\},\{b,d\}}|+|\bbx_{\{a,d\},\{b,c\}}|\right)\\
&+\sum_{\ell\in [m]}\sum_{A\in\binom{[m]\setminus\{\ell\}}{3}}|\bbx_{A}|.
\end{align*}
Through double counting arguments, we have
\begin{align*}
    \sum_{\ell\in [m]}\sum_{\{a,b,c,d\}\in\binom{[m]\setminus \{\ell\}}{4}}\left(|\bbx_{\{a,b\},\{c,d\}}|+|\bbx_{\{a,c\},\{b,d\}}|+|\bbx_{\{a,d\},\{b,c\}}|\right)=(m-4)M_1
\end{align*}
and
\begin{align*}
    \sum_{\ell\in [m]}\sum_{A\in\binom{[m]\setminus\{\ell\}}{3}}|\bbx_{A}|&=(m-3)M_2.
\end{align*}
Combined with (\ref{eq00b_(n,N,k,k+2)}), this implies (\ref{eq00c_(n,N,k,k+2)}).

Similarly, through double counting arguments, (\ref{eq00d_(n,N,k,k+2)}) follows by summing both sides of (\ref{eq00a1_(n,N,k,k+2)}) over all $\{a,b\}\in {[m]\choose 2}$, and (\ref{eq00e_(n,N,k,k+2)}) follows by summing both sides of (\ref{eq00a2_(n,N,k,k+2)}) over all $\{a,b,c\}\in {[m]\choose 3}$.

Moreover, by summing both sides of (\ref{eq00a3_(n,N,k,k+2)}) over all $\{a,b,c\}\in {[m]\choose 3}$, we have
\begin{align}
    &\sum_{\{a,b,c\}\in \binom{[m]}{3}}(|\bbx_{\{a,b\},\{c\}}|+|\bbx_{\{a,c\},\{b\}}|+|\bbx_{\{b,c\},\{a\}}|)=\nonumber\\
    &\sum_{\{a,b,c\}\in \binom{[m]}{3}}\sum_{\{a',b'\}\in\binom{[m]\setminus \{a,b,c\}}{2}}\left(|\bbx_{\{a,b\},\{a',b'\}}|+|\bbx_{\{a,c\},\{a',b'\}}|+|\bbx_{\{b,c\},\{a',b'\}}|\right). \label{eq00a4_(n,N,k,k+2)}
\end{align}
Note that for each $\{d,e,f,g\}\in {[m]\choose 4}$, $|\bbx_{\{d,e\},\{f,g\}}|$ appears in the RHS of (\ref{eq00a4_(n,N,k,k+2)}) if and only if $\{d,e\}\subseteq \{a,b,c\}$ and $\{f,g\}=\{a',b'\}$, or $\{f,g\}\subseteq \{a,b,c\}$ and $\{d,e\}=\{a',b'\}$. Thus, for each $\{d,e,f,g\}\in {[m]\choose 4}$, $|\bbx_{\{d,e\},\{f,g\}}|$ appears exactly $2(m-4)$ times in the RHS of (\ref{eq00a4_(n,N,k,k+2)}). Thus, by the definition of $M_1$, the RHS of (\ref{eq00a4_(n,N,k,k+2)}) is equal to $2(m-4)M_1$, which leads to (\ref{eq00f_(n,N,k,k+2)}).
\end{proof}


Armed with the identities in Corollary \ref{coro: LB for m=k+2}, we now present the proof of Theorem \ref{Thm: LB of $(n,N,k,k+2)$-PIRA}. Similar to the proofs of Theorem \ref{Thm: general LB_1} and Theorem \ref{Thm: general LB_2}, we also utilize information entropy and Lemma \ref{lem: size of $m-k+1$ buckets} to proceed with the proof. However, to obtain a better bound than Theorem \ref{Thm: general LB_2}, we employ the response matrix and Lemma \ref{lem: property of response matrix} to obtain a more accurate characterization of the non-information symbols in each bucket.

\begin{proof}[Proof of Theorem \ref{Thm: LB of $(n,N,k,k+2)$-PIRA}]
Let $\cC$ be an $(n,N,k,m)$-PIR array code with $N=N_p(n,k,k+2)$ and $m=k+2$.

Fix some $\ell\in [m]$. Consider the bucket $\bbc_\ell$ which contains $N_{\ell}$ symbols. Assume that $|\bbx_{\{\ell\}}|=n_{\ell}$ for some $1 \leq n_\ell \leq N_\ell \leq n$. Then, according to Lemma \ref{lem: property of response matrix}, by reordering the indices of the information symbols if necessary, we may assume that
$$\bbc_\ell=(\bbx_{\{\ell\}},\bby_{\ell}),$$
where $\bby_{\ell}$ is a vector of length $N_\ell-n_\ell$ such that $\bby_\ell(s)=\bbc_\ell(n_\ell+s)$, $s\in [N_\ell-n_\ell]$, is the response from the bucket $\bbc_{\ell}$ in the recovery process of some information symbol $x_{i_{l,s}}\in \bbx$. Then, by Lemma \ref{lem: LB for m=k+2}, for each $s\in [N_{\ell}-n_{\ell}]$, either $x_{i_{l,s}}\in \bbx_{\{\ell,a,b\}}$ for some $\{a,b\}\in \binom{[m]\setminus \{\ell\}}{2}$, or
$x_{i_{l,s}}\in \bbx_{\{\ell,a\},\{b,c\}}$ for some $\{a,b,c\}\in \binom{[m]\setminus \{\ell\}}{3}$. For every $\{a,b\}\in \binom{[m]\setminus \{\ell\}}{2}$ and every $\{a,b,c\}\in \binom{[m]\setminus \{\ell\}}{3}$, we denote $\bby_{\ell,\{a,b\}}$ as the projection of $\bby_{\ell}$ to $\{s\in [N_{\ell}-n_{\ell}]:x_{i_{l,s}}\in \bbx_{\{\ell,a,b\}}\}$ and $\bby_{\ell,a,\{b,c\}}$ as the projection of $\bby_{\ell}$ to $\{s\in [N_{\ell}-n_{\ell}]:x_{i_{l,s}}\in \bbx_{\{\ell,a\},\{b,c\}}\}$. 
Then, we have
\begin{equation}\label{eq02_(n,N,k,k+2)}
|\bby_\ell|=\sum_{\{a,b\}\in \binom{[m]\setminus \{\ell\}}{2}}|\bby_{\ell,\{a,b\}}|+\sum_{\{a,b,c\}\in \binom{[m]\setminus \{\ell\}}{3}}\left(|\bby_{\ell,a,\{b,c\}}|+|\bby_{\ell,b,\{a,c\}}|+|\bby_{\ell,c,\{a,b\}}|\right).
\end{equation}

Recall that $\sum_{\ell\in [m]}|\bbc_{\ell}|=\sum_{\ell\in [m]}(|\bbx_{\{\ell\}}|+|\bby_{\ell}|)=N$. Hence, (\ref{eq02_(n,N,k,k+2)}) implies that
\begin{align}\label{eq01_(n,N,k,k+2)}
\sum_{\ell\in [m]}\sum_{\{a,b,c\}\in \binom{[m]\setminus \{\ell\}}{3}}&(|\bby_{\ell,a,\{b,c\}}|+|\bby_{\ell,b,\{a,c\}}|+|\bby_{\ell,c,\{a,b\}}|)
=\sum_{\ell\in [m]}\left(|\bby_\ell|-\sum_{\{a,b\}\in \binom{[m]\setminus \{\ell\}}{2}}|\bby_{\ell,\{a,b\}}|\right)\nonumber\\
=&N-\sum_{\ell\in [m]}|\bbx_{\{\ell\}}|-\sum_{\ell\in [m]}\sum_{\{a,b\}\in \binom{[m]\setminus \{\ell\}}{2}}|\bby_{\ell,\{a,b\}}|\nonumber\\
\overset{(\mathrm{I})}{\geq}&N-\sum_{\ell\in [m]}|\bbx_{\{\ell\}}|-\sum_{\ell\in [m]}\sum_{\{a,b\}\in \binom{[m]\setminus \{\ell\}}{2}}|\bbx_{\{\ell,a,b\}}|\nonumber\\
\overset{(\mathrm{II})}{=}& N-(m-4)n-4M_2,
\end{align}
where $(\mathrm{I})$ follows from $|\bby_{\ell,\{a,b\}}|\leq |\bbx_{\{\ell,a,b\}}|$ and $(\mathrm{II})$ follows from $\sum_{\ell\in [m]}\sum_{\{a,b\}\in \binom{[m]\setminus \{\ell\}}{2}}|\bbx_{\{\ell,a,b\}}|=3M_2$ and (\ref{eq00c_(n,N,k,k+2)}).

Next, let $\bbx$ be uniformly distributed over $\mathbb{F}_q^n$ and we establish the lower bound using the information entropy. 

For any $\{a,b,c\}\in \binom{[m]}{3}$, by Lemma \ref{lem: size of $m-k+1$ buckets}, $\{a,b,c\}$ contains a recovery set for every $x_i\in \bbx$. Thus, we have
\begin{align}
n=&H(\bbc_a,\bbc_b,\bbc_c)=H(\bbc_a)+H(\bbc_b|\bbc_a)+H(\bbc_c|\bbc_a,\bbc_b)\nonumber\\
\overset{(\mathrm{III})}{=}&H(\bbc_a)+H(\bbx_{\{b\}},\bby_{b}|\bbc_a)+H(\bbx_{\{c\}},\bby_{c}|\bbc_a,\bbc_b)\nonumber\\
=&H(\bbc_a)+H(\bbx_{\{b\}}|\bbc_a)+H(\bby_{b}|\bbc_a,\bbx_{\{b\}})+H(\bbx_{\{c\}}|\bbc_a,\bbc_b)+H(\bby_{c}|\bbc_a,\bbc_b,\bbx_{\{c\}}) \nonumber\\
\overset{(\mathrm{IV})}{\leq}&H(\bbc_a)+|\bbx_{\{b\}}|-|\bbx_{\{a\},\{b\}}|+|\bby_{b}|+H(\bbx_{\{c\}}|\bbc_a,\bbc_b)+H(\bby_{c}|\bbc_a,\bbc_b,\bbx_{\{c\}}) \nonumber\\
\overset{(\mathrm{V})}{\leq}&H(\bbc_a)+|\bbx_{\{b\}}|-|\bbx_{\{a\},\{b\}}|+|\bby_b|+ \nonumber\\
&|\bbx_{\{c\}}|-H(\bbx_{\{a\},\{c\}},\bbx_{\{b\},\{c\}},\bbx_{\{a,b\},\{c\}})+H(\bby_{c}|\bbc_a,\bbc_b,\bbx_{\{c\}})\nonumber\\
\overset{(\mathrm{VI})}{=}&H(\bbc_a)+|\bbx_{\{b\}}|-|\bbx_{\{a\},\{b\}}|+|\bby_b|+ \nonumber\\
&|\bbx_{\{c\}}|-|\bbx_{\{a\},\{c\}}|-|\bbx_{\{b\},\{c\}}|-|\bbx_{\{a,b\},\{c\}}|+|\bbx_{\{a\},\{b\},\{c\}}|+H(\bby_{c}|\bbc_a,\bbc_b,\bbx_{\{c\}})\nonumber\\
\overset{(\mathrm{VII})}{\leq}&|\bbc_a|+|\bbc_b|+|\bbx_{\{c\}}|-(|\bbx_{\{a\},\{b\}}|+|\bbx_{\{a\},\{c\}}|+|\bbx_{\{b\},\{c\}}|-|\bbx_{\{a\},\{b\},\{c\}}|)\nonumber\\
       &-|\bbx_{\{a,b\},\{c\}}|+H(\bby_{c}|\bbc_a,\bbc_b,\bbx_{\{c\}}), \label{eq1_(n,N,k,k+2)}
\end{align}
where $(\mathrm{III})$ follows from the assumption that $\bbc_\ell=(\bbx_{\{\ell\}},\bby_{\ell})$ for each $\ell\in [m]$, $(\mathrm{IV})$ follows from
$$H(\bbx_{\{b\}}|\bbc_a)\leq H(\bbx_{\{b\}}|\bbx_{\{a\},\{b\}})= |\bbx_{\{b\}}|-|\bbx_{\{a\},\{b\}}|$$
and $H(\bby_{b}|\bbc_a,\bbx_{\{b\}})\leq |\bby_b|$, $(\mathrm{V})$ follows from
\begin{align*}
  H(\bbx_{\{c\}}|\bbc_a,\bbc_b)&\leq H(\bbx_{\{c\}}|\bbx_{\{a\},\{c\}},\bbx_{\{b\},\{c\}},\bbx_{\{a,b\},\{c\}})\\
  &= |\bbx_{\{c\}}|-H(\bbx_{\{a\},\{c\}},\bbx_{\{b\},\{c\}},\bbx_{\{a,b\},\{c\}}),
\end{align*}
$(\mathrm{VI})$ follows from the conclusion of Lemma \ref{lem: LB for m=k+2} that $\bbx_{\{a,b\},\{c\}}$ is independent with $\bbx_{\{a\},\{c\}}$ and $\bbx_{\{b\},\{c\}}$ and
\begin{align*}
  H(\bbx_{\{a\},\{c\}},\bbx_{\{b\},\{c\}},\bbx_{\{a,b\},\{c\}})&= |\bbx_{\{a,b\},\{c\}}|+|\bbx_{\{a\},\{c\}}|+|\bbx_{\{b\},\{c\}}|-I(\bbx_{\{a\},\{c\}};\bbx_{\{b\},\{c\}})\\
  &=|\bbx_{\{a,b\},\{c\}}|+|\bbx_{\{a\},\{c\}}|+|\bbx_{\{b\},\{c\}}|-|\bbx_{\{a\},\{b\},\{c\}}|,
\end{align*}
$(\mathrm{VII})$ follows from $|\bbc_{\ell}|=|\bbx_{\{\ell\}}|+|\bby_{\ell}|$.

Next, we upper bound $H(\bby_{c}|\bbc_a,\bbc_b,\bbx_{\{c\}})$ by the following claim.
\begin{claim}\label{claim_(k,k+2)}
$H(\bby_{c}|\bbc_a,\bbc_b,\bbx_{\{c\}})\leq |\bby_c|-\sum_{\{d,e\}\in \binom{[m]\setminus\{a,b,c\}}{2}}(|\bby_{c,a,\{d,e\}}|+|\bby_{c,b,\{d,e\}}|)$.
\end{claim}
\begin{proof}[Proof of Claim \ref{claim_(k,k+2)}]
We only need to show that $\bbc_a,\bbc_b,\bbx_{\{c\}}$ together can recover $\bby_{c,a,\{d,e\}}$ and $\bby_{c,b,\{d,e\}}$ for any $\{d,e\}\in \binom{[m]\setminus\{a,b,c\}}{2}$. Fix a $y\in \bby_{c,a,\{d,e\}}$. Recall that $y$ is the response from $\bbc_{c}$ in recovery of some symbol $x\in \bbx_{\{a,c\},\{d,e\}}$. By $\{d,e\}\subseteq [m]\setminus\{a,b,c\}$, we know that: 1) $x\in \bbx_{\{b\}}$; 2) $x$ can be recovered by $y$ and $z$, where $z$ is the response from $\bbc_a$ in recovery of $x$. Thus, $y$ can be recovered by $x$ and $\bbc_a$. By $x\in \bbx_{\{b\}}$ and $\bbc_b=(\bbx_{\{b\}},\bby_b)$, this shows that $\bby_{c,a,\{d,e\}}$ can be recovered by $\bbc_a,\bbc_b$. Similarly, we can show that $\bby_{c,b,\{d,e\}}$ can also be recovered by $\bbc_a,\bbc_b$. This concludes the result.
\end{proof}

Now, by Claim \ref{claim_(k,k+2)} and $|\bbc_i|=N_i$, (\ref{eq1_(n,N,k,k+2)}) implies that
\begin{align*}
n\leq & N_a+N_b+N_c-(|\bbx_{\{a\},\{b\}}|+|\bbx_{\{a\},\{c\}}|+|\bbx_{\{b\},\{c\}}|-|\bbx_{\{a\},\{b\},\{c\}}|) \nonumber\\
&-|\bbx_{\{a,b\},\{c\}}|-\sum_{\{d,e\}\in \binom{[m]\setminus\{a,b,c\}}{2}}(|\bby_{c,a,\{d,e\}}|+|\bby_{c,b,\{d,e\}}|)
\end{align*}
holds for any $\{a,b,c\}\in \binom{[m]}{3}$. This leads to
\begin{align}\label{eq4_(n,N,k,k+2)}
3n\leq & 3(N_a+N_b+N_c)-3(|\bbx_{\{a\},\{b\}}|+|\bbx_{\{a\},\{c\}}|+|\bbx_{\{b\},\{c\}}|-|\bbx_{\{a\},\{b\},\{c\}}|)\nonumber\\
&-(|\bbx_{\{a,b\},\{c\}}|+|\bbx_{\{a,c\},\{b\}}|+|\bbx_{\{b,c\},\{a\}}|)\nonumber\\
&-\sum_{\{d,e\}\in \binom{[m]\setminus\{a,b,c\}}{2}}(|\bby_{c,a,\{d,e\}}|+|\bby_{c,b,\{d,e\}}|+|\bby_{a,b,\{d,e\}}|+|\bby_{a,c,\{d,e\}}|+|\bby_{b,a,\{d,e\}}|+|\bby_{b,c,\{d,e\}}|).
\end{align}
By summing both sides of (\ref{eq4_(n,N,k,k+2)}) over all $\{a,b,c\}\in \binom{[m]}{3}$, we have
\begin{align}\label{eq5_(n,N,k,k+2)}
3n\binom{m}{3}\leq & \binom{m-1}{2} \sum_{i\in [m]}3N_i-3\left((m-2)\sum_{\{a,b\}\in \binom{[m]}{2}}|\bbx_{\{a\},\{b\}}|-\sum_{\{a,b,c\}\in \binom{[m]}{3}}|\bbx_{\{a\},\{b\},\{c\}}|\right)\nonumber\\
&-\sum_{\{a,b,c\}\in \binom{[m]}{3}}(|\bbx_{\{a,b\},\{c\}}|+|\bbx_{\{a,c\},\{b\}}|+|\bbx_{\{b,c\},\{a\}}|)\nonumber\\
&-\sum_{\{a,b,c\}\in \binom{[m]}{3}}\sum_{\{d,e\}\in \binom{[m]\setminus\{a,b,c\}}{2}}(|\bby_{c,a,\{d,e\}}|+|\bby_{c,b,\{d,e\}}|+|\bby_{a,b,\{d,e\}}|\nonumber\\
&~~~~~~~~~~~~~~~~~~~~~~~~~~~~~~~~~~~~~+|\bby_{a,c,\{d,e\}}|+|\bby_{b,a,\{d,e\}}|+|\bby_{b,c,\{d,e\}}|).
\end{align}

Meanwhile, through double counting, we have
\begin{align*}
    \sum_{\{a,b,c\}\in \binom{[m]}{3}}\sum_{\{d,e\}\in \binom{[m]\setminus\{a,b,c\}}{2}}&(|\bby_{c,a,\{d,e\}}|+|\bby_{c,b,\{d,e\}}|+|\bby_{a,b,\{d,e\}}|+|\bby_{a,c,\{d,e\}}|+|\bby_{b,a,\{d,e\}}|+|\bby_{b,c,\{d,e\}}|)\\
    =(m-4)&\sum_{\ell\in [m]}\sum_{\{a,b,c\}\in \binom{[m]\setminus \{\ell\}}{3}}(|\bby_{\ell,a,\{b,c\}}|+|\bby_{\ell,b,\{a,c\}}|+|\bby_{\ell,c,\{a,b\}}|).
\end{align*}
By plugging this equality together with (\ref{eq00d_(n,N,k,k+2)}), (\ref{eq00e_(n,N,k,k+2)}), (\ref{eq00f_(n,N,k,k+2)}) from Corollary \ref{coro: LB for m=k+2} and $N=\sum_{i\in [m]}N_i$ into (\ref{eq5_(n,N,k,k+2)}), we have
\begin{align}\label{eq6_(n,N,k,k+2)}
3n\binom{m}{3}\leq & 3N\binom{m-1}{2}-\left(m(m-4)(m-5)M_1+(m-\frac{1}{2})(m-3)(m-4)M_2\right)\nonumber\\
&-2(m-4)M_1\nonumber\\
&-(m-4)\sum_{a\in [m]}\sum_{\{b,c,d\}\in \binom{[m]\setminus \{a\}}{3}}(|\bby_{a,b,\{c,d\}}|+|\bby_{a,c,\{b,d\}}|+|\bby_{a,d,\{b,c\}}|)\nonumber\\
\overset{(\mathrm{VIII})}{\leq} & 3N\binom{m-1}{2}-\left(m(m-4)(m-5)M_1+(m-\frac{1}{2})(m-3)(m-4)M_2\right)\nonumber\\
&-2(m-4)M_1-(m-4)\left(N-(m-4)n-4M_2\right)\nonumber\\
= & (\frac{3}{2}m^2-\frac{11}{2}m+7)N+(m-4)^2n-(m-4)\left((m^2-5m+2)M_1+(m^2-\frac{7}{2}m-\frac{5}{2})M_2\right)\nonumber\\
\overset{(\mathrm{IX})}{=} & (\frac{3}{2}m^2-\frac{11}{2}m+7)N-(m-4)(m^2-6m+6)n-\frac{3}{2}(m-3)(m-4)M_2,
\end{align}
where $(\mathrm{VIII})$ follows from (\ref{eq01_(n,N,k,k+2)}) and $(\mathrm{IX})$ follows from $n=M_1+M_2$. Finally, since
$M_2\geq 0$ and $m=k+2$, (\ref{eq6_(n,N,k,k+2)}) implies that
$$N\geq (k-2+\frac{4k+16}{3k^2+k+4})n.$$

This completes the proof.
\end{proof}

\section{Constructions of batch array codes with low redundancy}\label{sec: constructions}

In this section we provide three different ways to design BACs, where the first two are explicit constructions and the last one is a random construction.

\subsection{A construction through cyclic shifted sets}\label{subsec: constructions1}

\begin{construction}\label{Con: $(n, N=(2k-m+(m-k)^2/k)n, k, m)$-BAC}
Let $n$, $k$ and $m$ be positive integers such that $k\mid n$, $k<m< 2k$, and $(m-k) \mid k$. For $\ell\in [k]$, define
\begin{equation}\label{eq1-Thm:$(n, N=(2k-m+(m-k)^2/k)n, k, m)$-BAC}
    P_\ell=\{(\ell-1+a)\frac{n}{k}+b \pmod{n}:0\leq a\leq m-k-1,1\leq b\leq \frac{n}{k}\}.
\end{equation}
Clearly, each $P_\ell$ is a cyclic shift of $P_1$ by an interval of length $(\ell-1)\frac{n}{k}$. We call $\{P_1,\ldots,P_k\}$ a family of cyclic shifted sets of $P_1$. Then, for $\bbx\in\mathbb{F}_q^n$, define $\cC(\bbx)=(\bbc_1,\bbc_2,\ldots, \bbc_m)$  as follows:
\begin{itemize}
    \item For each $\ell\in [k]$, let $$\bbc_\ell=\bbx|_{[n]\backslash P_\ell}.$$
    \item For each $\ell\in [k+1, m]$, let
    \begin{align*}
        \bbc_{\ell} = \sum_{t=0}^{\frac{k}{m-k}-1}(x_{t(m-k) \frac{n}{k}+1}, x_{t(m-k) \frac{n}{k}+2}, \ldots, x_{t(m-k) \frac{n}{k}+(m-k)\frac{n}{k}}).
    \end{align*}
\end{itemize}

Note that the buckets $\bbc_{k+1},\ldots,\bbc_{m}$ are identical.
\end{construction}

To show that the code obtained in Construction \ref{Con: $(n, N=(2k-m+(m-k)^2/k)n, k, m)$-BAC} is indeed a BAC, we first prove the following lemma.
\begin{lemma}\label{lem: $(n, N=(2k-m+(m-k)^2/k)n, k, m)$-BAC}
    Let $\cC$ be the code obtained by Construction \ref{Con: $(n, N=(2k-m+(m-k)^2/k)n, k, m)$-BAC}. For each $\ell\in [k]$ and $i\in [n]$, if $x_i\notin \bbc_{\ell}$, then $x_i$ can be recovered by $\{\bbc_{\ell},\bbc_{\ell'}\}$ for any $\ell'\in [k+1,m]$.
\end{lemma}
\begin{proof}
    As $\bbc_{k+1},\ldots,\bbc_{m}$ are identical, we just need to show that $x_i$ can be recovered by the symbols from $\bbc_{\ell}$ and $\bbc_{k+1}$.  Assume that $i \equiv b \pmod{(m-k)\frac{n}{k}}$ with $1\leq b\leq (m-k)\frac{n}{k}$, and let $a=\frac{(i-b)k}{(m-k)n}$, i.e., $i = a(m-k)\frac{n}{k} + b$. By Construction \ref{Con: $(n, N=(2k-m+(m-k)^2/k)n, k, m)$-BAC}, $x_i\notin \bbc_\ell$ implies that $i \in P_\ell$. For simplicity, we only prove the case when $\ell\leq 2k-m+1$, the proof of the other case is similar. In this case, we have
    $$\frac{n}{k}(\ell-1)+1\leq a(m-k)\frac{n}{k} + b\leq \frac{n}{k}(\ell-1) + \frac{n}{k}(m-k).$$
    As $|P_\ell|=(m-k)\frac{n}{k}$, for any $t\in [0, \frac{k}{m-k}-1]\setminus\{a\}$, we have $t(m-k)\frac{n}{k} + b \notin P_\ell$. This implies that $x_{t(m-k)\frac{n}{k}+b}\in \bbc_\ell$ for $t\in [0, \frac{k}{m-k}-1]\setminus\{a\}$. Note that $\sum_{t=0}^{\frac{k}{m-k}-1}x_{t(m-k)\frac{n}{k}+b}\in \bbc_{k+1}$, and thus we can recover $x_i$ by the symbols from $\bbc_\ell$ and $\bbc_{k+1}$.
\end{proof}

\begin{theorem}\label{Thm: $(n, N=(2k-m+(m-k)^2/k)n, k, m)$-BAC}
    The code $\cC$ obtained by Construction \ref{Con: $(n, N=(2k-m+(m-k)^2/k)n, k, m)$-BAC} is an $(n, N=(2k-m+\frac{(m-k)^2}{k})n, k, m)$-BAC.
\end{theorem}
\begin{proof}
     By Definition \ref{def_BAC}, it suffices to show that for any multi-set of requests $\ms{i_1,\ldots,i_k}$, there are $k$ mutually disjoint subsets $R_1,\ldots,R_k$ of $[m]$ such that for each $j\in [k]$, $R_j$ is a recovery set of $x_{i_j}$.

     First note that, due to the cyclic structures of $P_1,\dots,P_k$, each index $i\in [n]$ is covered exactly $m-k$ times by these sets, and thus exactly $m-k$ buckets within ${\bbc_1, \bbc_2, \ldots, \bbc_k}$ do not contain $x_{i}$. Define $G=(V, E)$ as the bipartite graph with vertex set $V=A\cup B$, where $A = [2k-m]$ and $B=[k]$. For any $j\in [2k-m]$ and $\ell \in [k]$, $(j, \ell)\in E$ if and only if $x_{i_j}\in \bbc_{\ell}$. Since each vertex in $A$ has degree $k-(m-k)=2k-m$, by Hall's marriage theorem\cite{Hall1935OnRO}, there is a complete matching from $A$ to $B$. Thus, there are $2k-m$ different buckets $\bbc_{\ell_1}, \bbc_{\ell_2}, \ldots, \bbc_{\ell_{2k-m}}$ such that $x_{i_j}\in\bbc_{\ell_j}$ for each $j\in [2k-m]$. In this way we find the recovery sets for the first $2k-m$ requests and there are $m-k$ unused buckets among ${\bbc_1, \bbc_2, \ldots, \bbc_k}$.

     Next we construct $x_{i_j}$'s recovery sets for $2k-m+1\leq j \leq k$. Arbitrarily choose a one-to-one correspondence between $\{x_{i,j}:2k-m+1\leq j \leq k\}$ and the set of unused $m-k$ buckets. The bucket corresponding to $x_{i_j}$ can either handle the request itself, or handel the request with the help of a bucket within ${\bbc_{k+1}, \ldots, \bbc_m}$ according to Lemma \ref{lem: $(n, N=(2k-m+(m-k)^2/k)n, k, m)$-BAC}.

     Thus the code $\cC$ obtained by Construction \ref{Con: $(n, N=(2k-m+(m-k)^2/k)n, k, m)$-BAC} is indeed an $(n, N, k, m)$-BAC, where the computation of $N$ is straightforward.
\end{proof}

\begin{remark}
    By Theorem \ref{Thm: general LB_2} and Theorem \ref{Thm: $(n, N=(2k-m+(m-k)^2/k)n, k, m)$-BAC}, for $m=k+1$, the code obtained from Construction \ref{Con: $(n, N=(2k-m+(m-k)^2/k)n, k, m)$-BAC} is an $(n,(k-1+\frac{1}{k})n,k,k+1)$-BAC with optimal code length.
\end{remark}

Theorem \ref{Thm: $(n, N=(2k-m+(m-k)^2/k)n, k, m)$-BAC} provides an upper bound on $N_{B}(n, k, m)$ for $k<m<2k$. Note that as a function of $m$,
$$2k-m+\frac{(m-k)^2}{k} = \frac{1}{k}(m-\frac{3k}{2})^2 + \frac{3k}{4}$$
achieves its minimum when $m=\frac{3k}{2}$. Hence, due to the fact that $N_{B}(n, k, m)\leq N_{B}(n, k, m')$ for any $m\geq m'$, we have the following direct corollary from Theorem \ref{Thm: $(n, N=(2k-m+(m-k)^2/k)n, k, m)$-BAC}.
\begin{corollary}\label{coro: UB on N_B by Cons 1}
    When $m \leq \frac{3k}{2}$ and $(m-k)|k$,
    \begin{align*}
        N_{P}(n, k, m)\leq N_{B}(n, k, m)\leq (2k-m+\frac{(m-k)^2}{k})n,
    \end{align*}
    When $m> \frac{3k}{2}$ and $(m-k)|k$,
    \begin{align*}
        N_{P}(n, k, m)\leq N_{B}(n, k, m)\leq N_{B}(n, k, \lfloor\frac{3k}{2}\rfloor)\leq \frac{3k}{4}n + \frac{n}{k}(\lfloor\frac{3k}{2}\rfloor-\frac{3k}{2})^2.
    \end{align*}
\end{corollary}

In \cite{IKOS04} and \cite{RSDG16}, a gadget lemma was used to construct general batch codes through batch codes with small parameters. Similar to their approach, we have the following variant of the gadget lemma in \cite{IKOS04}, which helps us to construct PIR array codes and BACs for large parameters.

\begin{lemma}[Gadget lemma]\label{lem: gadget lemma}
    The following hold for both PIR array codes and BACs:
    \begin{itemize}
        \item If there exists an $(n, N_1, k_1, m_1)$-BAC (PIR array code) and an $(n, N_2, k_2, m_2)$-BAC (PIR array code), then there is an $(n, N_1+N_2, k_1+k_2, m_1+m_2)$-BAC (PIR array code).
        \item If there exists an $(n_1, N_1, k_1, m_1)$-BAC (PIR array code) and an $(n_2, N_2, k_2, m_2)$-BAC (PIR array code), then there is an $(n_1+n_2, N_1+N_2, \min\{k_1,k_2\}, m_1+m_2)$-BAC (PIR array code).
        \item If there exists an $(n, N, k, m)$-BAC (PIR array code), then for any positive integer $c$, there is an $(cn, cN, k, m)$-BAC (PIR array code).
    \end{itemize}
\end{lemma}
\begin{proof}
    We only prove the first term, the proofs of the other two are similar.

    Let $\cC_1$ be an $(n, N_1, k_1, m_1)$-BAC (PIR array code) and $\cC_2$ be an $(n, N_2, k_2, m_2)$-BAC (PIR array code). For $\bbx\in \mathbb{F}_q^n$, let $\cC_1(\bbx) = (\bbc^{(1)}_1,\ldots,\bbc^{(1)}_{m_1})$ and $\cC_2(\bbx) = (\bbc^{(2)}_1,\ldots,\bbc^{(2)}_{m_2})$. Define $\cC(\bbx) = (\bbc^{(1)}_1,\ldots,\bbc^{(1)}_{m_1}, \bbc^{(2)}_1,\ldots,\bbc^{(2)}_{m_2})$. Next, we show that $\cC$ is an $(n, N_1+N_2, k_1+k_2, m_1+m_2)$-BAC (PIR array code).

    Fix any multiset of requests $\ms{i_1,i_2,\ldots, i_{k_1+k_2}}$. As $\cC_1$ is an $(n, N_1, k_1, m_1)$-BAC (PIR array code), we can obtain $k_1$ mutually disjoint recovery sets  $S_1,\ldots,S_{k_1}$ for $\ms{i_1,i_2,\ldots, i_{k_1}}$ such that $S_j\subseteq [m_1]$ for $j\in [k_1]$. As $\cC_2$ is an $(n, N_2, k_2, m_2)$-BAC (PIR array code), we can obtain $k_2$ mutually disjoint recovery sets $T_1,\ldots,T_{k_2}$ for $\ms{i_{k_1+1},i_{k_1+2},\ldots, i_{k_1+k_2}}$ such that  $T_j\subseteq [m_2]$ for $j\in [k_2]$. Thus, for $\ms{i_1,\ldots, i_{k_1+k_2}}$, we take $R_j=S_j$ for $j\in [k_1]$ and $R_j=\{t+m_1:t\in T_j\}$ for $j\in [k_1+1,k_1+k_2]$ as the recovery set for $x_{i_j}$ corresponding to the code $\cC$. Clearly, $R_1,\ldots,R_{k_1+k_2}$ are mutually disjoint subsets in $[m_1+m_2]$.
\end{proof}

By the gadget lemma, we have the following immediate corollary.

\begin{corollary}
    For any positive integers $n, k, m, c$ with $k \leq m$, the following hold:
    \begin{itemize}
        \item $N_B(n, ck, cm)\leq cN_B(n, k, m)$ and $N_P(n, ck, cm)\leq cN_P(n, k, m)$.
        \item $N_B(cn, k, cm)\leq cN_B(n, k, m)$ and $N_P(cn, k, cm)\leq cN_P(n, k, m)$.
        \item $N_B(cn, k, m)\leq cN_B(n, k, m)$ and $N_P(cn, k, m)\leq cN_P(n, k, m)$.
    \end{itemize}
\end{corollary}

Let $\cC$ be the BAC obtained in Construction \ref{Con: $(n, N=(2k-m+(m-k)^2/k)n, k, m)$-BAC}. Then, the first $k$ buckets of $\cC$ have size $n(1-\frac{m-k}{k})$ and the last $m-k$ buckets have size $\frac{(m-k)n}{k}$. Thus, when $\frac{m-k}{k}\neq \frac{1}{2}$, $\cC$ is not a uniform BAC. Next, we modify Construction \ref{Con: $(n, N=(2k-m+(m-k)^2/k)n, k, m)$-BAC} and provide a construction of uniform BAC for $m=k+1$.

\begin{construction}\label{Con1_uniform}
Let $n,k$ be positive integers such that $k(k+1)\mid n$. Let $\cC_0$ be the $(\frac{n}{k+1}, (k-1+\frac{1}{k})\frac{n}{k+1}, k, k+1)$-BAC obtained by Construction \ref{Con: $(n, N=(2k-m+(m-k)^2/k)n, k, m)$-BAC}. For any $\bbx=(x_1,x_2,\ldots,x_n)\in\Fq^n$ and $j\in[k+1]$, denote
$$\bbx_j = (x_{(j-1)\frac{n}{k+1} + 1}, \ldots, x_{\frac{jn}{k+1}})$$
and $\cC_0(\bbx_j) = (\bbc_{j,1},\ldots, \bbc_{j,k+1})$. Then, we define $\cC(\bbx) = (\bbc_1, \bbc_2,\ldots, \bbc_{k+1})$ as follows:
    $$\bbc_{\ell} = (\bbc_{1,\ell},\bbc_{2,\ell-1}, \ldots, \bbc_{k+1,\ell-k}),$$
where the addition of indices are modular $k+1$.
\end{construction}

\begin{example}[Example \ref{ex1.1} continued]
Let $\cC_0$ be the $(n=4, N=13, k=4, m=5)$-BAC given in Example  \ref{ex1.1}. Then, $\cC_0$ is an illustration of Construction \ref{Con: $(n, N=(2k-m+(m-k)^2/k)n, k, m)$-BAC} for $m=k+1$. Next, based on $\cC_0$, we construct a uniform $(n=20, N=65,  k=4, m=5)$-BAC through Construction \ref{Con1_uniform}.

Given $\bbx\in \Fq^{20}$, let $\bbx_j = (x_{4j-3},x_{4j-2},x_{4j-1}, x_{4j})$ for $j\in[5]$. Then, we have $\cC_0(\bbx_j) = (\bbc_{j, 1}, \ldots, \bbc_{j, 5})$ with
\begin{align*}
    &\bbc_{j, 1} = (x_{4j-3}, x_{4j-2}, x_{4j-1}),\\
    &\bbc_{j, 2} = (x_{4j-3}, x_{4j-2}, x_{4j}),\\
    &\bbc_{j, 3} = (x_{4j-3}, x_{4j-1}, x_{4j}),\\
    &\bbc_{j, 4} = (x_{4j-2}, x_{4j-1}, x_{4j}),\\
    &\bbc_{j, 5} = x_{4j-3}+x_{4j-2}+x_{4j-1}+x_{4j}.
\end{align*}
Then, one can check that the code $\cC$ defined as follows is an $(n=20, N=65,  k=4, m=5)$-BAC.
    \begin{table}[h]
        \centering
        \begin{tabular}{|c|c|c|c|c|}
        \hline
        $\bbc_1$ & $\bbc_2$ & $\bbc_3$ & $\bbc_4$ & $\bbc_5$ \\\hline
        $(x_{1}, x_{2}, x_{3})$ & $(x_{1}, x_{2}, x_{4})$ & $(x_{1}, x_{3}, x_{4})$ & $(x_{2}, x_{3}, x_{4})$ & $x_{1} + x_{2} + x_{3} + x_{4}$ \\\hline
        $x_{5} + x_{6} + x_{7} + x_{8}$ & $(x_{5}, x_{6}, x_{7})$ & $(x_{5}, x_{6}, x_{8})$ & $(x_{5}, x_{7}, x_{8})$ & $(x_{6}, x_{7}, x_{8})$\\\hline
        $(x_{10}, x_{11}, x_{12})$ & $x_{9} + x_{10} + x_{11} + x_{12}$ & $(x_{9}, x_{10}, x_{11})$ & $(x_{9}, x_{10}, x_{12})$ & $(x_{9}, x_{11}, x_{12})$\\\hline
        $(x_{13}, x_{15}, x_{16})$ & $(x_{14}, x_{15}, x_{16})$ & $x_{13} + x_{14} + x_{15} + x_{16}$ & $(x_{13}, x_{14}, x_{15})$ & $(x_{13}, x_{14}, x_{16})$\\\hline
        $(x_{17}, x_{18}, x_{20})$ & $(x_{17}, x_{19}, x_{20})$ & $(x_{18}, x_{19}, x_{20})$ & $x_{17} + x_{18} + x_{19} + x_{20}$ & $(x_{17}, x_{18}, x_{19})$\\\hline
        \end{tabular}
    \end{table}

\end{example}

\begin{theorem}\label{Thm: uniform-BAC}
The code $\cC$ obtained in Construction \ref{Con1_uniform} is a uniform $(n, N=(k-1+\frac{1}{k})n, k, k+1)$-BAC with each bucket having size $(k-1+\frac{1}{k})\frac{n}{k+1}$.
\end{theorem}

\begin{proof}
For each $j\in [k+1]$ and each $i\in [(j-1)\frac{n}{k+1} + 1, \frac{jn}{k+1}]$, by Construction \ref{Con: $(n, N=(2k-m+(m-k)^2/k)n, k, m)$-BAC}, there is only one bucket of $\bbc_{j,1},\ldots,\bbc_{j,k}$ which doesn't contain $x_i$ and this bucket together with $\bbc_{j,k+1}$ can recover $x_i$.

Note that for each $\ell\in [k+1]$ and $j\in [k+1]$, only $\bbc_{j,l+j-1}$ of $\bbc_{j,1},\ldots,\bbc_{j,k+1}$ is contained in $\bbc_{\ell}$. Therefore, for each $i\in [(j-1)\frac{n}{k+1} + 1, \frac{jn}{k+1}]$, there are exactly $k-1$ buckets in $\{\bbc_1, \bbc_2,\ldots, \bbc_{k+1}\}$ contains $x_i$ and the remaining two buckets can recover $x_i$. Similar to the proof of Theorem \ref{Thm: $(n, N=(2k-m+(m-k)^2/k)n, k, m)$-BAC}, for any multi-set request $\ms{{i_1}, \ldots, {i_k}}$, we can choose the $k-1$ buckets which contain $x_{i_1}, \ldots, x_{i_{k-1}}$, respectively, to recover $x_{i_1}, \ldots, x_{i_{k-1}}$,  and use the remaining two buckets to recover $x_{i_k}$.

Moreover, since $\{\ell,\ell+1,\ldots,\ell+k\}=[k+1]$ for each $\ell$, we have $|\bbc_{\ell}|=(k-1+\frac{1}{k})\frac{n}{k+1}$. That is, each bucket of $\cC$ has the same size and $\cC$ is a uniform BAC.
\end{proof}

\begin{remark}\label{rmk-uniform-BAC}
In \cite{NY22}, the authors studied uniform BAC under a more generalized setting: the number of symbols being accessed in each bucket during the recovery of each $x_{i_j}$ is limited. 
In view of this setting, by Lemma \ref{lem: $(n, N=(2k-m+(m-k)^2/k)n, k, m)$-BAC} and Theorem \ref{Thm: uniform-BAC}, Construction \ref{Con1_uniform} provides a uniform BAC where the number of symbols being accessed in each bucket during the recovery of each $x_{i_j}$ is at most $k-1$.
\end{remark}

\subsection{A construction through ``good vectors''}\label{subsec: constructions2}

Next, we will construct a class of PIR array codes (which are also BACs for some parameters) for $m = n\leq 2k$. We start with the following example.

\begin{example}[$(n=5, N=10, k=3, m=5)$-BAC]\label{ex: $(n=5, N=10, k=3, m=5)$-computational batch code}
For $\bbx=(x_1, x_2, \ldots, x_5)$, define $\cC(\bbx)=(\bbc_1,\ldots,\bbc_5)$ as follows:
\begin{table}[h]
    \centering
    \begin{tabular}{|c|c|c|c|c|}
    \hline
    $\bbc_1$&$ \bbc_2$&$\bbc_3$&$\bbc_4$&$\bbc_5$\\
    \hline
    $x_1 $& $x_2$ & $x_3$ & $x_4$ & $x_5$ \\
    \hline
    $x_3+x_4$ & $x_4+x_5$& $x_5+x_1$ & $x_1+x_2$ &$x_2+x_3$\\
    \hline
    \end{tabular}
\end{table}

To see that $\cC$ is a $(5,10,3,5)$-BAC, it suffices to show that for any multiset request $\ms{i_1,i_2,i_3}$, there are $3$ mutually disjoint subsets $R_1,R_2,R_3\subseteq [m]$ such that $R_j$ is the recovery set for $i_j$, $j\in [3]$.

When $i_1, i_2, i_3$ are distinct, we can take the $R_1,R_2,R_3$ as $\{i_1\}, \{i_2\}, [5]\setminus\{i_1,i_2\}$.

When $i_1=i_2=i_3=i$ for some $i\in [5]$, we assume w.l.o.g. that $i=1$. As $x_2\in\bbc_2$ and $x_1+x_2\in\bbc_4$, thus, $x_1$ can be recovered by $\{\bbc_2, \bbc_4\}$. Similarly, $x_1$ can be recovered by $\{\bbc_3, \bbc_5\}$. As $x_1\in \bbc_1$, $\{1\}$, $\{2,4\}$ and $\{3,5\}$ are the $3$ disjoint recovery sets of $x_1$.

When $\ms{i_1,i_2,i_3}=\ms{i,i,i'}$ for some distinct $i,i'\in [5]$. Note that $x_i$ can be recovered by $\bbc_i$, $\{\bbc_{i+1}, \bbc_{i-2}\}$, $\{\bbc_{i+2},\bbc_{i-1}\}$. As $i'\neq i$, one of $\{{i+1}, {i-2}\}$ and $\{{i+2}, {i-1}\}$ is contained in $[5]\setminus \{i,i'\}$. Thus, the recovery sets for request $\ms{i,i,i'}$ are $\{i\},[5]\setminus \{i,i'\},\{i'\}$.
\end{example}

Next, we extend the code in Example \ref{ex: $(n=5, N=10, k=3, m=5)$-computational batch code} to general parameters.

\begin{definition}\label{Def: good vectors}
    Let $t$ be a positive integer. We call $\bbv = (v_1, v_2 ,\ldots, v_{2t})\in [t]^{2t}$ or $\bbv = (v_1, v_2 ,\ldots, v_{2t+1})\in [0, t]^{2t+1}$ a $good$ vector with respect to (w.r.t.) $t$, if the following holds:
    \begin{itemize}
        \item [1.] Any integer in $[t]$ appears exactly twice in $\bbv$.
        \item [2.] If $v_i=v_{i'}=j$ for some $j\in [t]$ and $1\leq i < i'$, then $i'-i=j$.
    \end{itemize}
    For a good vector $\bbv$ w.r.t. $t$, we define $j(\bbv)=\max\{i:~v_i=j\}$.
\end{definition}

\begin{example}
We have the following examples of good vectors:
\begin{itemize}
    \item $(1, 1)$ is the only good vector w.r.t $t=1$ of length $2$.
    \item $\bbv=(2, 3, 2, 4, 3, 1, 1, 4)$ is a good vector w.r.t $t=4$ of length $8$. Moreover, we have  $1(\bbv)=7$, $2(\bbv)=3$, $3(\bbv)=5$, and $4(\bbv)=8$.
    \item There is no good vector w.r.t $t=2$ and $t=3$ of length $2t$.
\end{itemize}
\end{example}

\begin{lemma}\label{lem: good vector of length 2t+1}
    For any positive integer $t$, there is a good vector w.r.t $t$ of length $2t+1$.
\end{lemma}
\begin{proof}
Define
$$\bbv_1=\begin{cases}
    (t-1,t-3,\ldots,3,1,1,3,\ldots,t-3,t-1),~\text{if $2\mid t$};\\
    (t,t-2,\ldots,3,1,1,3,\ldots,t-2,t),~\text{if $2\nmid t$},
\end{cases}$$
and
$$\bbv_2=\begin{cases}
    (t,t-2,\ldots,4,2,0,2,4,\ldots,t-2,t),~\text{if $2\mid t$};\\
    (t-1,t-3,\ldots,4,2,0,2,4,\ldots,t-3,t-1),~\text{if $2\nmid t$}.
\end{cases}$$
Next, we show that $\bbv=(\bbv_1,\bbv_2)$ is a good vector w.r.t $t$. We only prove the case when $2\mid t$, the proof of the other case is similar.

    Clearly, every integer in $[t]$ appears exactly twice in $\bbv$. Thus, the first condition holds. Note that for each odd $j\in\{1, 3, \ldots, t-1\}$, $v_i = j$ if and only if $i\in \{\frac{t-j+1}{2}, \frac{t+j+1}{2}\}$. Thus, if $v_i = v_{i'} = j$ for some odd $j$ and $1\leq i\leq i'$, then $i'-i=j$. Similarly, for each even $j\in \{2, 4, \ldots, t\}$, $v_i = j$ if and only if $i\in \{t+1+\frac{t-j}{2}, t+1+\frac{t+j}{2}\}$. Thus, if $v_i = v_{i'} = j$ for some even $j$ and $1\leq i\leq i'$, then $i'-i=j$. This verifies the second condition.
\end{proof}

Now we introduce a construction based on the good vectors defined above. The addition and subtraction of indices are modulo $n$.

\begin{construction}\label{Con: $(n, N = (t+1)(4t+1), k = 2t+1, m=n)$-CPIR}
Let $\bbv$ be a good vector w.r.t $t$. Let $n = 4t+1$ if $|\bbv|=2t$, and $n=4t+2$ if $|\bbv|=2t+1$. For each $i\in [n]$ and $j\in [t]$, define
$$y_{i,j}=x_{i-t-j(\bbv)} + x_{i-t-j(\bbv)+j}.$$
Then, we define the codeword encoded from $\bbx$, $\cC(\bbx)=(\bbc_1,\bbc_2,\ldots, \bbc_n)$, as follows. For each $i\in [n]$,
$$\bbc_i = (x_i, y_{i,1},y_{i,2},\ldots,y_{i,t}).$$
\end{construction}
\begin{remark}
     Example \ref{ex: $(n=5, N=10, k=3, m=5)$-computational batch code} is the code obtained from Construction \ref{Con: $(n, N = (t+1)(4t+1), k = 2t+1, m=n)$-CPIR} with $\bbv = (1, 1)$.
\end{remark}
\begin{example}
    Let $\bbv = (1, 1, 2, 0, 2)$ and $\cC$ be the code obtained from Construction $\ref{Con: $(n, N = (t+1)(4t+1), k = 2t+1, m=n)$-CPIR}$ based on $\bbv$. One can check that $\bbv$ is a good vector w.r.t $t=2$ of length $5$. Moreover, $1(\bbv) = 2$ and $2(\bbv) = 5$. Let $n = 4 t + 2 = 10$ and $\bbx = (x_1, x_2, \cdots, x_{10})$. Then, the codeword $\cC(\bbx)$ has the following form:
    \begin{table}[h]
        \centering
        \begin{tabular}{|c|c|c|c|c|c|c|c|c|c|}
        \hline
        $\bbc_1$&$ \bbc_2$&$\bbc_3$&$\bbc_4$&$\bbc_5$ & $\bbc_6$&$ \bbc_7$&$\bbc_8$&$\bbc_9$&$\bbc_{10}$\\
        \hline
        $x_1$&$ x_2$&$x_3$&$x_4$&$x_5$ & $x_6$&$ x_7$&$x_8$&$x_9$&$x_{10}$\\
        \hline
        $x_7+x_8 $& $x_8+x_9$ & $x_9+x_{10}$ & $x_{10}+x_1$ & $x_1+x_2$ & $x_2+x_3$& $x_3+x_4$ & $x_4+x_5$ & $x_5+x_6$ & $x_6+x_7$ \\
        \hline
        $x_4+x_6$ & $x_5+x_7$& $x_6+x_8$ & $x_7+x_9$ &$x_8+x_{10}$ & $x_9+x_1$ & $x_{10}+x_2$& $x_1+x_3$ & $x_2+x_4$ &$x_3+x_5$\\
        \hline
        \end{tabular}
    \end{table}
\end{example}


\begin{lemma}\label{lem: Cons2 is CPIR}
    The code obtained in Construction \ref{Con: $(n, N = (t+1)(4t+1), k = 2t+1, m=n)$-CPIR} is an $(n, N = (t+1)n, k = 2t+1, m = n)$-PIR array code.
\end{lemma}
\begin{proof}
    Let $\bbv$ be a good vector w.r.t $t$.
    We only prove the case when $|\bbv|=2t+1$, the proof for the other case when $|\bbv|=2t$ is similar. Let $\cC$ be the code obtained from Construction \ref{Con: $(n, N = (t+1)(4t+1), k = 2t+1, m=n)$-CPIR} based on $\bbv$. Next, we show that there are $2t+1$ mutually disjoint recovery sets for each information symbol $x_i$.

    Clearly, $x_i$ can be recovered by $\{x_i\}$ and
    $$\begin{array}{cc}
    \{x_{i-1},x_{i-1}+x_{i}\},  &  \{x_{i+1},x_{i+1}+x_{i}\}, \\
    \{x_{i-2},x_{i-2}+x_{i}\},  &  \{x_{i+2},x_{i+2}+x_{i}\}, \\
    \vdots & \\
    \{x_{i-t},x_{i-t}+x_{i}\},  &  \{x_{i+t},x_{i+t}+x_{i}\}.
    \end{array}$$
    Note that for each $j\in [t]$, $x_{i-j}\in \bbc_{i-j}$, $x_{i+j}\in \bbc_{i+j}$,
    $$x_{i-j}+x_i=y_{i+t+j(\bbv)-j,j}\in \bbc_{i+t+j(\bbv)-j},$$ and
    $$x_i + x_{i+j}=y_{i+t+j(\bbv),j}\in \bbc_{i+t+j(\bbv)}.$$
    Next, we show that $A=\{i-t,\ldots,i+t\}$, $B=\{i+t+j(\bbv)-j\}_{j\in [t]}$ and $C=\{i+t+j(\bbv)\}_{j\in [t]}$ are disjoint subsets of $[n]$. Then, we can take the $2t+1$ mutually disjoint recovery sets of $x_i$ as $\{i\}$, $\{i-j,i+t+j(\bbv)-j\}_{j\in [t]}$ and $\{i+j,i+t+j(\bbv)\}_{j\in [t]}$.

    We first show that $A\cap B = \varnothing$. Assume that there are $j_1 \in [-t, t]$ and $j_2\in [t]$ such that $i + j_1 \equiv i + t + j_2(\bbv) - j_2  \pmod n$. Then, $t + j_2(\bbv) - j_2 - j_1 \equiv 0 \pmod n$. By Definition \ref{Def: good vectors}, $j_2(\bbv)$ is the maximal index such that $v_{j_2(\bbv)} = j_2$. As $j_2$ appears exactly twice in $\bbv$, by the second property of the good vector, $j_2(\bbv) - j_2$ is $\min\{j:~v_j=j_2\}$. This implies that $1\leq j_2(\bbv) - j_2\leq 2t$. As $j_1\in [-t, t]$, thus, we have $$1\leq t + j_2(\bbv) - j_2 - j_1\leq 4t.$$
    This contradicts the fact that $n = 4t+2$ and $t + j_2(\bbv) - j_2 + j_1 \equiv 0 \pmod n$.

    Next, we show that $A\cap C = \varnothing$. Similarly, we assume that there are $j_1 \in [-t, t]$ and $j_2\in [t]$ such that $i + j_1 \equiv i + t + j_2(\bbv)  \pmod n$. Then, $t + j_2(\bbv) - j_1 \equiv 0 \pmod n$. Clearly, $2\leq j_2(\bbv)\leq 2t+1$. Thus, by $j_1\in [-t,t]$, we have
    $$2\leq t + j_2(\bbv) - j_1\leq 4t+1.$$
    This contradicts the fact that $n = 4t+2$ and $t + j_2(\bbv) - j_1 \equiv 0 \pmod n$.

    Finally, we show that $B\cap C = \varnothing$. Otherwise, we assume that there are $j_1, j_2\in [t]$ such that $i + t + j_1(\bbv) \equiv i + t + j_2(\bbv) - j_2  \pmod n$. Then, $j_2(\bbv) - j_2 - j_1(\bbv) \equiv 0 \pmod n$. As $1\leq j_2(\bbv) - j_2\leq 2t$ and $2\leq j_1(\bbv)\leq 2t+1$, $j_2(\bbv) - j_2 - j_1(\bbv) \equiv 0 \pmod n$ implies that $j_2(\bbv) - j_2 = j_1(\bbv)$. However, by the second property of good vector, we have $\bbv_{j_2(\bbv) - j_2} = j_2$ while $\bbv_{j_1(\bbv)} = j_1$. Thus, we have $j_1 = j_2$, which implies that $j_1(\bbv)=j_2(\bbv)$. By $j_2(\bbv) - j_2 = j_1(\bbv)$, this leads to $j_2 = 0$, a contradiction. This concludes the proof.
\end{proof}

By Lemma \ref{lem: good vector of length 2t+1}, Lemma \ref{lem: Cons2 is CPIR} and the gadget lemma \ref{lem: gadget lemma}, we have the following corollary.

\begin{corollary}\label{coro4.3}
    For any positive integer $t$, the following hold:
    \begin{itemize}
        \item If there is a good vector $\bbv$ w.r.t $t$ of length $2t$, then for any $n$ such that $(4t+1)\mid n$, $N_P(n, k=2t+1, m=4t+1)\leq (t+1)n$.
        \item For any $n$ such that $(4t+2)\mid n$, $N_P(n, k=2t+1, m=4t+2)\leq (t+1)n$.
    \end{itemize}
\end{corollary}

Next, we provide a sufficient condition on $k$ that guarantees the code obtained from Construction \ref{Con: $(n, N = (t+1)(4t+1), k = 2t+1, m=n)$-CPIR} is an $(n,(t+1)n,k,n)$-BAC for $n=4t+1$ and $n=4t+2$.

\begin{theorem}\label{Thm: sufficient cond for Cons2 being BAC}
    Let $k$ be a positive integer such that $2k\leq 2t+\Delta+\lceil\frac{k}{\Delta}\rceil$ holds for every $\Delta\in [k]$. Then, the code obtained from Construction \ref{Con: $(n, N = (t+1)(4t+1), k = 2t+1, m=n)$-CPIR} is an $(n, N=(t+1)\bcom{n}, k, n)$-BAC.
\end{theorem}
\begin{proof}
    Let $\bbv$ be a good vector w.r.t $t$ and let $\cC$ be the code obtained from Construction \ref{Con: $(n, N = (t+1)(4t+1), k = 2t+1, m=n)$-CPIR} based on $\bbv$. In the following, we assume that $|\bbv| = 2t+1$, the proof for the case when $|\bbv|=2t$ is similar. It suffices to show that for any multi-set request $I=\ms{i_1,\ldots,i_k}\subseteq [n]$, there are $k$ mutually disjoint subsets $R_1,\ldots,R_k$ of $[m]=[n]$ such that for each $j\in [k]$, $R_j$ is a recovery set of $x_{i_j}$.

    Assume that $I$ consists of $\Delta$ different indices ${i_1}, {i_2}, \ldots, {i_{\Delta}}$ with multiplicities $a_1, a_2,\ldots, a_{\Delta}$. Then, $a_1+a_2+\cdots+a_{\Delta}=k$. W.l.o.g., we assume that $1\leq a_1\leq a_2\leq\cdots\leq a_{\Delta}$.

    By Lemma \ref{lem: Cons2 is CPIR}, for each $i\in [n]$, $x_{i}$ has $2t+1$ mutually disjoint recovery sets, $\{i\}$, $\{i-j, i+t+j(\bbv)-j\}_{j\in [t]}$ and $\{i+j, i+t+j(\bbv)\}_{j\in [t]}$. Denote $\cR_{i}=\{\{i\}$, $\{i-j, i+t+j(\bbv)-j\}_{j\in [t]},\{i+j, i+t+j(\bbv)\}_{j\in [t]}\}$ as the family of the $2t+1$ recovery sets for $x_{i}$. Next, based on $\cR_{i_j}$ for each $j\in [\Delta]$, we construct $a_j$ recovery sets for each $x_{i_j}$ through the following approach:
    \begin{itemize}
        \item First, for each $j\in [\Delta]$, since $x_{i_j}\in \bbc_{i_j}$, we take $\{i_j\}$ as the first recovery set of $x_{i_j}$.
        \item Next, assume that for a certain $1\leq j<\Delta$, we have already obtained $a_1+\cdots+a_{j-1}$ recovery sets for $x_{i_{1}},\ldots,x_{i_{j-1}}$, and all these recovery sets together with $\{i_j\},\{i_{j+1}\},\ldots,\{i_{\Delta}\}$ are mutually disjoint. Then, we choose arbitrarily $a_{j}-1$ recovery sets of $x_{i_j}$ from $\cR_{i_j}$ that are disjoint with all the constructed recovery sets. Note that each of these $a_j-1$ recovery sets has form $\{i_j-j', i_j+t+j'(\bbv)-j'\}$ or $\{i_j+', i_j+t+j'(\bbv)\}$ for some $j'\in [t]$.
    \end{itemize}

    Now, we show that the above approach can always proceed as long as $j\leq \Delta$.

    Otherwise, assume we fail at step $\delta$ ($\delta\leq \Delta$). That is, we have already obtained $a_1+\cdots+a_{\delta-1}$ recovery sets for $x_{i_{1}},\ldots,x_{i_{\delta-1}}$ and these recovery sets together with $\{i_\delta\},\{i_{\delta+1}\},\ldots,\{i_{\Delta}\}$ are mutually disjoint, but at least $2t - (a_{\delta}-1) + 1$ members of $\cR_{i_\delta}\setminus\{i_{\delta}\}$ intersect with the union of the $a_1+\cdots+a_{\delta-1}+\Delta-\delta+1$ constructed recovery sets.

    Note that each constructed recovery sets of size $1$ intersects at most $1$ set in $\cR_{i_\delta}\setminus\{i_{\delta}\}$, and each constructed recovery sets of size $2$ intersects at most $2$ sets in $\cR_{i_\delta}\setminus\{i_{\delta}\}$. Moreover, the constructed recovery sets consists of $\Delta$ recovery sets of size $1$, known as $\{i_1\},\ldots,\{i_{\Delta}\}$ and $a_1+\cdots+a_{\delta-1}-(\delta-1)$ recovery sets of size $2$. Thus, as $\{i_{\delta}\}$ is disjoint with other sets in $\cR_{i_\delta}\setminus\{i_{\delta}\}$, the number of sets in $\cR_{i_\delta}\setminus\{i_{\delta}\}$ intersect with the union of the constructed recovery sets is at most
    $$\Delta-1+2(a_1+\cdots+a_{\delta-1}-(\delta-1)).$$
    By the assumption, we have
    \begin{align*}
    \Delta-1+2(a_1+\cdots+a_{\delta-1}-(\delta-1))&\geq 2t - (a_{\delta}-1) + 1\\
    &\geq 2t - (a_{\Delta}-1) + 1.
    \end{align*}
    This implies that
    \begin{align*}
       0\leq &\Delta-1+2(a_1+\cdots+a_{\delta-1}-(\delta-1)) - (2t - (a_{\Delta}-1) + 1)\\
       = & \Delta - 1 + 2(a_1+\cdots+a_{\delta-1}+a_{\Delta}-\delta) - (2t+1)- (a_{\Delta}-1)\\
       \leq & \Delta - 1 + 2(k-\Delta) - (2t+a_{\Delta})\\
        = & 2k - 2t - \Delta - a_{\Delta}-1\\
        \leq &2k - 2t - \Delta - \lceil\frac{k}{\Delta}\rceil-1 < 0,
    \end{align*}
     which contradicts to the condition of $k$.

     This concludes the proof.
\end{proof}

By Theorem \ref{Thm: sufficient cond for Cons2 being BAC}, we have the following corollaries.

\begin{corollary}
    The code obtained from Construction \ref{Con: $(n, N = (t+1)(4t+1), k = 2t+1, m=n)$-CPIR} with $\bbv = (1, 1)$, which is also the code shown in Example \ref{ex: $(n=5, N=10, k=3, m=5)$-computational batch code}, is an $(n=5, N=10, k=3, m=5)$-BAC and the code obtained from Construction \ref{Con: $(n, N = (t+1)(4t+1), k = 2t+1, m=n)$-CPIR} with $\bbv = (2, 3, 2, 4, 3, 1, 1, 4)$ is an $(n=17, N=85, k=7, m=17)$-BAC.
\end{corollary}

\begin{corollary}
    The code obtained from Construction \ref{Con: $(n, N = (t+1)(4t+1), k = 2t+1, m=n)$-CPIR} is an $(n, N = (t+1)\bcom{n}, k = \lfloor(\sqrt{t+\frac{1}{4}}+\frac{1}{2})^2\rfloor, n)$-BAC.
\end{corollary}
\begin{proof}
    For any $\Delta\in [k]$, we have
    \begin{align*}
         2t+\Delta+\lceil\frac{k}{\Delta}\rceil-2k
         \geq &2t+2\sqrt{k}-2k\\
         = &-2(\sqrt{k}-\frac{1}{2})^2 + 2t + \frac{1}{2}.
    \end{align*}
    Thus, when $k = \lfloor(\sqrt{t+\frac{1}{4}}+\frac{1}{2})^2\rfloor$, we have $-2(\sqrt{k}-\frac{1}{2})^2 + 2t + \frac{1}{2}\geq 0$ and $2k\leq 2t+\Delta+\lceil\frac{k}{\Delta}\rceil$ holds for every $\Delta\in [k]$. Then, the result follows immediately.
\end{proof}

Using the gadget lemma, we combine Construction \ref{Con: $(n, N=(2k-m+(m-k)^2/k)n, k, m)$-BAC} and Construction \ref{Con: $(n, N = (t+1)(4t+1), k = 2t+1, m=n)$-CPIR} together and obtain the following upper bound on $N_p(n,k,m)$. This improves the result for $N_p(n,k,m)$ in Corollary \ref{coro: UB on N_B by Cons 1} when $m \leq \frac{3}{2}k$.

\begin{corollary}\label{coro: UB on N_P by Cons 2}
    When $\frac{3}{2}k < m < 2k$ satisfies $2m-3k$ is odd and $\mathrm{lcm}(4m-6k,4k-2m)\mid n$, we have
    \begin{align*}
        N_{P}(n, k, m)\leq \frac{3k-m+1}{2}n.
    \end{align*}
\end{corollary}
\begin{proof}
    Let $t = \frac{2m-3k-1}{2}$. Then, Construction \ref{Con: $(n, N = (t+1)(4t+1), k = 2t+1, m=n)$-CPIR} gives a $(4m-6k, \frac{2m-3k+1}{2}(4m-6k), 2m-3k, 4m-6k)$-PIR array code. By the gadget lemma \ref{lem: gadget lemma}, there is an $(n, \frac{2m-3k+1}{2}n, 2m-3k, 4m-6k)$-PIR array code for $n$ satisfying $(4m-6k)\mid n$. On the other hand, Construction \ref{Con: $(n, N=(2k-m+(m-k)^2/k)n, k, m)$-BAC} gives an $(n, \frac{3(2k-m)}{2}n, 4k-2m, 6k-3m)$-PIR array code for $n$ satisfying $(4k-2m)\mid n$. Thus, the first term of the gadget lemma \ref{lem: gadget lemma} implies that there is an $(n, \frac{3k-m+1}{2}n, k, m)$-PIR array code for $n$ satisfying $\mathrm{lcm}(4m-6k,4k-2m)\mid n$.
\end{proof}


\subsection{A random construction of systematic BACs}

Through a random construction based on point-line incidences on the plane $\mathbb{F}_q^2$, Polyanskaya et al. \cite{PPV20} obtained systematic linear batch codes with the smallest known redundancies for certain parameter regime. In this subsection, we modify their random construction and prove the following result.



\begin{theorem}\label{Thm: random_cons for BAC}
Let $q$ be a large enough prime power and $n=q^2$. Let $k$ and $s$ be positive integers such that $(ks)^{\frac{3}{2}}< \frac{n^{1/4}}{32\ln {n}}$. Then, there exists a binary $(n,N,k,2\lceil\frac{\sqrt{n}}{s}\rceil)$-BAC with length
$$N\leq n+64(ks)^{\frac{3}{2}}n^{\frac{3}{4}}\ln{n}.$$
\end{theorem}

First, we need some notations to present the construction of the code. Let $q$ be a prime power and set $n=q^2$. Let $(\cP,\mathcal{L})$ denote the finite affine plane of order $q$, where $\cP$ is the set of $n=q^2$ points and $\mathcal{L}$ is the set of all the $n+q$ lines over the plane. Then, any two lines in $\mathcal{L}$ intersect in at most one point, each line in $\mathcal{L}$ contains $q$ points, and each point in $\cP$ is contained in $q+1$ lines. For each $L\in \cL$, denote $\mathrm{Slope}(L)$ as the slope of $L$. Let $m=2\lceil\frac{q}{s}\rceil$ for some integer $s\in [q]$ to be determined. Denote $\mathcal{L}^{\infty}=\{L^{\infty}_1,L^{\infty}_2,\ldots,L^{\infty}_q\}$ as the set of $q$ lines with infinite slope in $\mathcal{L}$. For $i\in [\frac{m}{2}]$, denote $$\mathcal{L}^{i}=\{L\in \mathcal{L}:\mathrm{Slope}(L)= is-a ~\text{for some}~0\leq a\leq s-1\}.$$

\begin{construction}\label{random_con}
Let $0< p_1,p_2<1$ be two parameters to be determined later. 
\begin{itemize}
    \item[\underline{Step 1:}] Choose each line from $\cL\setminus \cL^{\infty}$ independently with probability $p_1$ (without repetition), and denote the resulting set of lines as $\mathcal{F}$.

    \item[\underline{Step 2:}] We operate on each line $L\in \cL\setminus \cL^{\infty}$ independently and construct a point set $P(L)\subseteq L$ by picking each point on $L$ independently with probability $p_2$.
\end{itemize}
Now, based on $(\cP,\mathcal{L})$ and $\mathcal{F}$, we construct a systematic linear code $\mathcal{C}$ of dimension $n$ as follows:
\begin{enumerate}
    \item Denote $\cP=\{x_1,x_2,\ldots,x_n\}$. Fix a bijection that associates $n$ information symbols with points in $\cP$. With a slight abuse of notation, we also use $x_i$ to denote the $i$-th symbol of the information vector $\bbx\in\Fq^n$.
    \item For every $L\in\mathcal{F}$, if $P(L)\neq\varnothing$, define a parity-check symbol $y_{L}=\sum_{x\in P(L)}x$.
    \item We define the codeword in $\cC$ that encodes $\bbx$, $\cC(\bbx)=(\bbc_1,\bbc_2,\ldots,\bbc_m)$, as follows\footnote{As $m = 2\lceil\frac{q}{s}\rceil$, $\frac{ms}{2}$ might be larger than $q$. Thus, we define $L_j^{\infty} = \varnothing$ for $j > q$ here.},
    $$\bbc_{\ell}=\begin{cases}
        (x_i,~x_i\in\bigcup_{j=(\ell-1)s+1}^{\ell s}L^{\infty}_j),~\text{ when $1\leq \ell\leq \frac{m}{2}$};\\
        (y_{L},~~L\in \mathcal{F}\cap \mathcal{L}^{\ell-\frac{m}{2}}),~\text{ when $\frac{m}{2}+1\leq \ell\leq m$}.
    \end{cases}$$
\end{enumerate}
\end{construction}

Denote $\cP_{\ell}=\bigcup_{j=(\ell-1)s+1}^{\ell s}L^{\infty}_j$. Then, for each $\ell\in [\frac{m}{2}]$, $\bbc_\ell$ consists of all the information symbols that lie in $\cP_\ell$, and $\bbc_{\ell+\frac{m}{2}}$ consists of all the parity-check symbols that correspond to lines in $\cL^{\ell}\cap \cF$. Note that as point sets, $\cP_1,\ldots,\cP_{\frac{m}{2}}$ form a partition of $\cP$. Thus, for each information symbol $x_i$, there is a unique $\ell\in [\frac{m}{2}]$ such that $x_i\in\bbc_\ell$. That is, $(\bbc_1,\ldots,\bbc_{\frac{m}{2}})$ is a permutation of $\bbx$ and $\sum_{\ell=1}^{\frac{m}{2}}|\bbc_\ell|=n$. Moreover, note that $\cL^{\ell}\cap \cL^{\ell'}=\varnothing$ for $\ell\neq \ell'$. Thus, for each $L\in \cF$, there is a unique $\ell\in[\frac{m}{2}]$ such that the parity-check symbol $y_L\in \bbc_{\ell+\frac{m}{2}}$. This leads to $\sum_{\ell=1}^{\frac{m}{2}}|\bbc_{\ell+\frac{m}{2}}|\leq |\cF|$. Therefore, the code $\cC$ obtained from Construction \ref{random_con} is indeed a systematic linear code of length $N\leq n+|\cF|$.

Notice that Construction \ref{random_con} may fail to yield an $(n,N,k,m)$-BAC. The construction fails if there exists a multiset of requests $I$ for which there are no mutually disjoint recovery sets.
Therefore, to prove Theorem \ref{Thm: random_cons for BAC}, we aim to demonstrate that with high probability (w.h.p.), the code $\cC$ resulting from Construction \ref{random_con} has the following property: For any multiset of requests $I=\ms{{i_1},{i_2},\ldots, {i_k}}$, there exists a recovery set for each $x_{i_j}$, and all these $k$ recovery sets are pairwise disjoint \footnote{It is noteworthy that the value of $k$ cannot surpass $m$, given that there are $m$ distinct buckets. Indeed, our assumption holds that $(ks)^{\frac{3}{2}}< \frac{n^{1/4}}{32\ln {n}}$. Since $n=q^2$ and $m=2\lceil\frac{q}{s}\rceil$, when $q$ is large enough, it follows that $k\leq \frac{n^{1/6}}{s}\leq m$.}.

To accomplish this objective, we employ an algorithm that given a multiset of requests, finds mutually disjoint recovery sets w.h.p.
\begin{enumerate}
    \item [1.] Let $I=\ms{i_1,i_2\ldots,i_k}$ be a multiset of requests. Assume that for a certain $1\leq j<k$, we have already established $j-1$ mutually disjoint recovery sets $R_{1},\ldots,R_{j-1}$ for $x_{i_1},\ldots,x_{i_{j-1}}$.
    \item [2.] To proceed with finding the subsequent recovery sets for $x_{i_j}$, consider all the lines in $\mathcal{F}$ that pass through $x_{i_j}$. Find a line $L_{j}\in \cF$ such that the following requirements hold:
    \begin{itemize}
        \item [1)] $x_{i_j}\in P(L_{j})$ and $L_{j}\cap \left(\{x_{i_1},\ldots,x_{i_k}\}\setminus\{x_{i_j}\}\right)=\varnothing$.
        \item [2)] $R_j$ is disjoint with $\bigcup_{j'=1}^{j-1}R_{j'}$, where
    \begin{align}\label{eq1_random_con}
    R_{j}=&\{\ell:~\ell\in [\frac{m}{2}]~\mathrm{s. t.}\left(P(L_{j})\setminus \{x_{i_j}\}\right)\cap \cP_\ell\neq \varnothing\} \nonumber\\
    &\cup\{\ell+\frac{m}{2}:~\ell\in [\frac{m}{2}]~\mathrm{s.t.}~L_{j}\in \cL^{\ell}\}.
    \end{align}
    \end{itemize}
    Since $y_{L_{j}}=\sum_{x\in P(L_{j})}x$, $x_{i_j}$ can be recovered by $y_{L_{j}}$ and symbols in $P(L_{j})\setminus\{x_{i_{j}}\}$. Hence, $x_{i_j}$ can be recovered by $\{\bbc_{\ell}:~\ell\in [\frac{m}{2}]~\mathrm{s.t.}~\left(P(L_{j})\setminus \{x_{i_j}\}\right)\cap \cP_\ell\neq \varnothing\}$ together with $\{\bbc_{\ell+\frac{m}{2}}:~\ell\in [\frac{m}{2}]~\mathrm{s.t.}~L_{j}\in \cL^{\ell}\}$. That is, $R_j$ is a recovery set for $x_{i_j}$.
    \item [3.] Repeat this process until $j=k$, i.e., until all the recovery sets are obtained.
\end{enumerate}

To prove the theorem, we need to estimate the probability that there exists a multiset of requests $I$ for which the algorithm above fails. For every multiset $I$ and every round $j$, the construction fails if at least one of the following occurs.
\begin{enumerate}
    \item There is no $L_{j}\in\cF$ such that $x_{i_j}\in P(L_j)$ and $L_{j}\cap \left(\{x_{i_1},\ldots,x_{i_k}\}\setminus\{x_{i_j}\}\right)=\varnothing$.
    \item The set $R_{j}$ defined in (\ref{eq1_random_con}) is not disjoint with all the previously chosen recovery sets.
    \item The redundancy of the code is larger than $O((ks)^{\frac{3}{2}}n^{\frac{3}{4}}\ln{n})$, i.e., $|\cF|>O((ks)^{\frac{3}{2}}n^{\frac{3}{4}}\ln{n})$.
 \end{enumerate}
Since the estimations of the failure probability of the algorithm are similar to those in the proof of Theorem 1 in \cite{PPV20}, we leave the detailed proof of Theorem \ref{Thm: random_cons for BAC} in the appendix for interested readers.


\section{Conclusion and further research}\label{sec: conclusion}

In this paper, inspired by the idea of trading computing powers for storage spaces as many others works in the field of storage codes, we study the batch array codes (BACs), which is a general array version of the batch codes introduced in \cite{IKOS04}. Under the setting of BACs, a server can respond to user requests by transmitting functions derived from the data it stored. With this property, BACs can have less storage overhead compared to the original batch codes for the same application scenarios.

Our focus in this paper lies in investigating the trade-off between $N$, the length of a BAC, and $k$, the parameter for disjoint-recovery-set property. We first establish a general lower bound on the code length, and then improved bounds are obtained for the cases when $k<m<2k$ and $m=k+2$. We also offer several code constructions. For some cases, the lengths of the codes obtained from these constructions match our lower bounds and thus prove the tightness of our bounds.

In the following, we list some questions that, in our opinion, could further improve the research on this topic.

\begin{enumerate}
    \item The first and foremost open question is constructions of good BACs. That is, BACs with small code length. There have been a lot of good constructions of PIR codes, batch codes and their variants in recent years using different methods. Some of these codes and their constructing methods may possibly be modified to obtain good BACs.
    \item The lower bound in Theorem \ref{Thm: general LB_2} is shown to be optimal for the case $m=k+1$ by Construction \ref{Con: $(n, N=(2k-m+(m-k)^2/k)n, k, m)$-BAC}. However, when $m=k+2$, it's no longer optimal, as shown by Theorem \ref{Thm: LB of $(n,N,k,k+2)$-PIRA}. We believe that, similar to the case $m=k+2$, the lower bound in Theorem \ref{Thm: general LB_2} is not optimal for general $m\geq k+3$. Hence, the next question is to prove a better lower bound for general $m\geq k+3$ and provide corresponding optimal constructions. In the proof of Theorem \ref{Thm: LB of $(n,N,k,k+2)$-PIRA}, properties of the response matrix are used in the analysis of the structure of the code. This provides a possible way to obtain better bounds for $m\geq k+3$. Moreover, hinted by the form of the lower bounds we derived, we conjecture that for $m=k+\delta$, where $\delta$ is a constant, the lower bound is of the form $N_P(n,k,k+\delta)\geq(k-1+\Omega(\frac{1}{k}))n$.
    \item In \cite{NY22}, the authors studied uniform BACs under the setting where the number of symbols being accessed in each bucket during the recovery of each $x_{i_j}$ is at most a given parameter. Hence, another possible direction for the study of non-uniform BACs is to consider codes under this setting.
\end{enumerate}

\section*{Acknowledgement}

We thank Prof. Itzhak Tamo for many illuminating discussions. We also thank Prof. Yuval Ishai for sharing the full version of their paper \cite{IKOS04} to us.

\appendices

\section{Proof of Theorem \ref{Thm: random_cons for BAC}}

To simplify the proof of Theorem \ref{Thm: random_cons for BAC}, we use several lemmas. The first lemma is the Chernoff bound.
\begin{lemma}[Chernoff bound \cite{Chernoff52}, see Thm. 4.4 and Thm. 4.5 in \cite{Prob2017} for the current version]
\label{lem:chern}
Let $X=\sum_{i=1}^{n}X_i$, where $X_i=1$ with probability $p_i$ and $X_i=0$ with
probability $1-p_i$, and all $X_i$ are independent. Let $\mu=\E[X]$. Then,
\begin{equation*}
    \begin{cases}
        \Pr(X\geq (1+\delta)\mu)\leq e^{-\frac{\delta^2\mu}{2+\delta}},~~~~ \forall~ \delta>0,\\
        \Pr(X\leq (1-\delta)\mu)\leq e^{-\mu\delta^2/2},~~\forall~ 0<\delta<1.
    \end{cases}
\end{equation*}
\end{lemma}

In Construction \ref{random_con}, we choose $\cF$ randomly and then use $\cF$ to produce the parity-check symbol of the constructed code $\cC$. To show that the length of the code is small, we first bound the size $|\cF|$.
\begin{lemma}
    \label{lem:help1}
    Let $(\cP,\mathcal{L})$ and $\mathcal{F}$ be chosen according to Construction \ref{random_con}. Let $0<p_1,p_2<1$ be two parameters to be determined later.
    Then
    \[\Pr\parenv{|\cF|\geq 2p_1 n}\leq e^{-\frac{p_1n}{3}}.\]
\end{lemma}

\begin{IEEEproof}
    Note that every line from $\cL\setminus \cL^{\infty}$ is chosen to $\cF$ independently with probability $p_1$.
    Thus, the size $|\cF|$ is binomially distributed $B(n,p_1)$.
    This implies that $\E\sparenv{|\cF|}=p_1n$.
    Using Chernoff bound with $\delta=1$, we have $\Pr\parenv{|\cF|\geq 2p_1n}\leq e^{-\frac{p_1n}{3}}$.
\end{IEEEproof}

The next lemma bounds the probability that Construction \ref{random_con} fails to generate an $(n,N,k,m)$-BAC. 
\begin{lemma}
    \label{lem:help2}
    Let $k$ and $s$ be integers such that $(ks)^{\frac{3}{2}}\leq \frac{n^{1/4}}{32\ln {n}}$. Consider Construction \ref{random_con} with the parameters $p_1=32\sqrt{\frac{(ks)^3}{q}}\ln{n}$ and $p_2=\frac{1}{2\sqrt{ksq}}$. Denote $A$ as the event that Construction \ref{random_con} fails to generate an $(n,N,k,2\lceil\frac{\sqrt{n}}{s}\rceil)$-BAC. Then
    \[\Pr(A)\leq ne^{-\frac{n^{1/6}}{4}}+e^{-\frac{k\ln {n}}{2}}.\]
\end{lemma}

\begin{IEEEproof}
    Let $m=2\lceil\frac{q}{s}\rceil$ and we assume w.l.o.g. that $s|q$. Recall that Construction \ref{random_con} fails to generate an $(n,N,k,m)$-BAC if there is a multiset of requests $I$ for which there are no mutually disjoint recovery sets.

    For $i\in [n]$, we define $A_{i,w}$ as the event that there exists a multiset $I=\ms{{i_1},\ldots, {i_k}}$ satisfying $i_w=i$ for some $w\in [k]$, such that the above algorithm finds mutually disjoint recovery sets for $x_{i_1},\dots,x_{i_{w-1}}$, but fails to find one desired recovery set for $x_{i_w}=x_i$. Let $C$ denote the event that there is a line $L\in\cF$ such that $|P(L)|>2p_2q$. We have that
    \begin{align}\label{eq2_random_con}
        \Pr(A)&\leq \Pr\parenv{\bigcup_{i\in [n]\atop w\in [k]} A_{i,w}}\nonumber\\
        &\leq \Pr(C)+\Pr\parenv{\bigcup_{i\in [n]\atop w\in [k]} (A_{i,w}\cap \overline{C})}\nonumber\\
        &\leq \Pr(C)+kn \max_{i\in [n]\atop w\in [k]}\Pr(A_{i,w}\cap \overline{C}).
    \end{align}


    We first estimate $\Pr(C)$. For a fixed line $L\in \cL\setminus \cL^{\infty}$, every point in $P(L)$ is chosen independently with probability $p_2$.
    Thus, $|P(L)|$ is distributed binomially $B(q,p_2)$ with $\E[|P(L)|]=p_2q$.
    By Chernoff bound (with $\delta=1$) we obtain that for every line $L\in \cF$,
    \[\Pr\parenv{|P(L)|\geq 2p_2q}\leq e^{-\frac{p_2q}{3}}.\]
    Plugging $p_2$ into $p_2q$ we get $p_2q=\frac{\sqrt{q}}{2\sqrt{ks}}=\frac{n^{1/4}}{2\sqrt{ks}}$.
    Using the assumption that $(ks)^{\frac{3}{2}}\leq \frac{n^{1/4}}{\ln {n}}$ we have $p_2q\geq \frac{n^{1/6}}{4}$.
    By the union bound, this implies that
    \[\Pr(C)\leq  \sum_{L\in \cL\setminus \cL^{\infty}}\Pr\parenv{|P(L)|\geq 2p_2q}\leq ne^{-\frac{n^{1/6}}{4}}.\]


    We now estimate the probability $\Pr(A_{i,w}\cap \overline{C})$. For a multiset $I=\ms{{i_1},\ldots, {i_k}}$ with $i_w=i$,
    let $D_2(I)$ denote the event that the algorithm manages to find a recovery sets for each $x_{i_j}$, $1\leq j<w$, such that these $w-1$ recovery sets are mutually disjoint, and let $D_1(I)$ denote the event that the algorithm fails to find a desired recovery set for $x_i$. From the definition of $A_{i,w}$ and the union bound, we have
    \[\Pr(A_{i,w}\cap \overline{C})\leq \sum_{I=\ms{{i_1},\ldots, {i_k}} \atop i_w=i}\Pr(D_1(I)\cap D_2(I)\cap \overline{C}).\]
    Note that the number of all possible multisets $I=\ms{{i_1},\ldots, {i_k}}$ satisfying $i_w=i$ is $n^{k-1}$. Thus, we have
    \begin{align}\label{eq22a}
    \Pr(A_{i,w}\cap \overline{C})&< n^{k-1}  \max_{I=\ms{{i_1},\ldots, {i_k}} \atop i_w=i}\Pr(D_1(I)\cap D_2(I)\cap \overline{C})\nonumber\\
    &=n^{k-1}  \max_{I=\ms{{i_1},\ldots, {i_k}} \atop i_w=i}\Pr(D_1(I) | D_2(I)\cap \overline{C}) \Pr(D_2(I)\cap \overline{C})\nonumber\\
    &\leq n^{k-1}  \max_{I=\ms{{i_1},\ldots, {i_k}} \atop i_w=i}\Pr(D_1(I) | D_2(I)\cap \overline{C}).
    \end{align}


    We are now left to estimate $\Pr(D_1(I) | D_2(I)\cap \overline{C})$. By the definition of $C$, conditioned on $\overline{C}$, for every $L\in {\cF}$, $|P(L)|< 2p_2q$. Moreover, conditioned on $D_2(I)$, the algorithm has found $w-1$ mutually disjoint recovery sets $R_{1},\ldots,R_{w-1}$ for $x_{i_1},\ldots,x_{i_{w-1}}$. For every $1\leq j\leq w-1$, denote $L_{j}$ as the line in $\cF$ using to construct the recovery set $R_{j}$ for $x_{i_j}$ in the algorithm. Then, $|P(L_{j})|<2p_2q$. Since $\cP_1,\ldots,\cP_{\frac{m}{2}}$ form a partition of $\cP$, we have
    $$\abs{\{u\in [\frac{m}{2}]: \cP_u\cap P(L_j)\neq \varnothing\}}< 2p_2q.$$
    Denote $U_1=\bigcup_{1\leq j<w}(R_{j}\cap [\frac{m}{2}])$ and $U_2=\bigcup_{1\leq j<w}(R_{j}\cap [\frac{m}{2}+1,m])$. By the definition of $R_{j}$ in (\ref{eq1_random_con}), we have
    \begin{align*}
        U_1&=\bigcup_{j\in [w-1]}\{u\in [\frac{m}{2}]:~\cP_u\cap\left(P(L_{j})\setminus \{x_{i_j}\}\right) \neq \varnothing\},\\
        U_2&=\bigcup_{j\in [w-1]}\{u+\frac{m}{2}: u\in [\frac{m}{2}]~\mathrm{s.t.}~L_{j}\in \cL^{u}\}.
    \end{align*}
    Thus, we have
    \begin{align}\label{eq22b}
    |U_1|&\leq \abs{\bigcup_{j\in[w-1]}\{u\in [\frac{m}{2}]: \cP_u\cap P(L_j)\neq \varnothing\}}< 2kp_2q.
    \end{align}
    Moreover, since $\cL^{1},\ldots,\cL^{\frac{m}{2}}$ is a partition of $\cL\setminus \cL^{\infty}$, for each $j\in[w-1]$, there is a unique $u\in [\frac{m}{2}]$ such that $L_{j}\in \cL^{u}$. Namely, $|\{u\in [\frac{m}{2}]:~L_{j}\in \cL^{u}\}|=1$ for each $j\in[w-1]$. Thus, we also have
    \begin{align}\label{eq22b2}
    |U_2|&\leq \sum_{j=1}^{w-1}|\{u+\frac{m}{2}:~u\in [\frac{m}{2}]~\mathrm{s.t.}~L_{j}\in \cL^{u}\}|< k.
    \end{align}
    For each $l\in [\frac{m}{2}]$, let $L_{i,l}$ be a line in $\mathcal{L}^{l}$ passing through $x_i$. Since $\cP_{u}$ consists of $s$ different lines with infinite slope, we have $|L_{i,l}\cap \cP_{u}|=s$ for every $l\in [\frac{m}{2}]$ and $u\in [\frac{m}{2}]$.
    Thus, by (\ref{eq22b}), we have
    \begin{equation}\label{eq22b3}
        \abs{L_{i,l}\cap \left(\bigcup_{u\in U_1} \cP_{u}\right)}< 2ksp_2q.
    \end{equation}
    Moreover, by (\ref{eq22b2}), we also have $\abs{\{l\in [\frac{m}{2}]: l+\frac{m}{2}\in U_2\}}<k$.
    Since $k\leq \frac{n^{1/6}}{s}$, this implies that
    $$\abs{\{l\in [\frac{m}{2}]: \cL^{l}\cap \{L_1,\ldots,L_{w-1}\}=\varnothing\}}>\frac{m}{2}-k\geq \lceil\frac{m}{4}\rceil.$$

    We assume w.l.o.g. that $4\mid m$ and $[\frac{m}{4}]\subseteq \{l\in [\frac{m}{2}]: \cL^{l}\cap \{L_1,\ldots,L_{w-1}\}=\varnothing\}$. Then, 
    we have
    \begin{equation}\label{eq22b4}
        [\frac{m}{2}+1,\frac{3m}{4}]\cap U_2=\varnothing.
    \end{equation}
    Let $\zeta_{1},\ldots,\zeta_{\frac{m}{4}}$ be binary random variables such that $\zeta_{l}=1$ if and only if $L_{i,l}\in \cF$, $x_i\in P(L_{i,l})$, and
    $P(L_{i,l})\cap \parenv{\bigcup_{u\in U_1} \cP_{u}}=\varnothing$.
    Notice that $\zeta_l$'s are independent random variables since each line is chosen independently, and each point from a chosen line is also chosen independently.
    Moreover, we have that
    \[\Pr(\zeta_l=1)\geq p_1p_2(1-p_2)^{2ksp_2q}.\]
    Indeed, $p_1$ is the probability that $L_{i,l}$ is selected to $\cF$, $p_2$ is the probability that $x_i$ is to $P(L_{i,l})$, and since (\ref{eq22b3}) shows that $L_{i,l}$ intersects $\bigcup_{u\in U_1} \cP_{u}$ with less then $2ksp_2q$ points, $(1-p_2)^{2ksp_2q}$ is the probability that no points from $L_{i,l}\cap \parenv{\bigcup_{u\in U_1} \cP_{u}}$ were selected.


    If $\sum_{l=1}^{\frac{m}{4}}\zeta_l\geq 1$, we know that there is an $l_0\in [\frac{m}{4}]$ such that $L_{i,l_0}\in \cF$, $x_i\in P(L_{i,l_0})$, and
    $P(L_{i,l_0})\cap \parenv{\bigcup_{u\in U_1} \cP_{u}}=\varnothing$. Then, we define
    \begin{align*}
        R_{i}=&\{u\in [\frac{m}{2}]:~\left(P(L_{i,l_0})\setminus \{x_{i}\}\right)\cap \cP^{u}\neq \varnothing\} \cup \{l_0+\frac{m}{2}\} 
    \end{align*}
    Note that symbols in $P(L_{i,l_0})\setminus\{x_{i}\}$ are contained in $(\bbc_{u}:~u\in R_i\cap [\frac{m}{2}])$ and $y_{L_{i,l_0}}$ is contained in $\bbc_{l_0+\frac{m}{2}}$.
    Thus, by $x_{i}=y_{L_{i,l_0}}-\sum_{x\in P(L_{i,l_0})\setminus\{x_{i}\}}x$, $R_i$ is a recovery set for $x_{i}$. On the other hand, by $P(L_{i,l_0})\cap \parenv{\bigcup_{u\in U_1} \cP_{u}}=\varnothing$, we have $R_i\cap U_1=\varnothing$ and by (\ref{eq22b4}), we also have $R_i\cap U_2=\varnothing$. These imply that
    $$R_i\cap (U_1\cup U_2)=R_i\cap\parenv{\bigcup_{1\leq j<w}R_{j}}=\varnothing.$$
    Therefore, if $\sum_{l=1}^{\frac{m}{4}}\zeta_l\geq 1$ then the construction of the desired recovery set for $x_i$ succeeds.

    Thus,
    \[\Pr(D_1(I) ~|~ D_2(I)\cap \overline{C})\leq \Pr(\sum_{l=1}^{\frac{m}{4}} \zeta_l <1).\]
    From Chernoff bound (Lemma \ref{lem:chern}) we obtain
    \begin{align}\label{eq23}
        \Pr(\sum_{l=1}^{\frac{m}{4}} \zeta_l <1)&\leq e^{-\frac{\sum_{l=1}^{\frac{m}{4}}\Pr(\zeta_l=1)}{2}\cdot\delta^2}\nonumber\\
        &\overset{(a)}{\leq} e^{-\frac{mp_1p_2(1-p_2)^{2ksp_2q}}{8}\cdot\delta^2},
    \end{align}
    where $\delta$ satisfies $(1-\delta)\sum_{l=1}^{\frac{m}{4}}\Pr(\zeta_l=1)=1$ and $(a)$ follows since $\Pr(\zeta_l=1)\geq p_1p_2(1-p_2)^{2ksp_2q}$.
    Using the inequality $(1-p)^x\geq 1-px$ which holds for $0\leq p\leq 1$ and $x\geq 0$, we obtain
    \[(1-p_2)^{2ksp_2q}\geq 1-2ksp_2^{2}q.\]
    Plugging in $p_2=\frac{1}{2\sqrt{ksq}}$ we obtain $1-2ksp_2^{2}q= \frac{1}{2}$.
    Plugging in $p_1=32\sqrt{\frac{(ks)^{3}}{q}}\ln{n}$, by $m=\frac{2q}{s}$, we have
    \[mp_1p_2(1-p_2)^{2ksp_2q}\geq \frac{qp_1p_2}{s}\geq 16k\ln{n}.\]
    Thus, for $q$ large enough (i.e., $n$ large enough), we obtain
    \begin{align*}
        \delta=1-\frac{1}{\sum_{l=1}^{\frac{m}{4}}\Pr(\zeta_l=1)}&\geq 1-\frac{4}{mp_1p_2(1-p_2)^{2ksp_2q}}\\
        &\geq 1-\frac{1}{4k\ln {n}}\geq 0.99.
    \end{align*}
    Putting everything together we obtain
    \[\Pr(D_1(I) ~|~ D_2(I)\cap C_j)\leq e^{-(2k\ln {n})\cdot\delta^2}\leq e^{-1.9k\ln {n}}.\]


    Substituting this upper bound in (\ref{eq22a}), we obtain
    \begin{align*}
        kn\cdot\max_{i\in [n]\atop w\in [k]}\Pr(A_{i,w}\cap \overline{C})& \leq kn^{k}\cdot e^{-1.9k\ln {n}}\leq e^{-\frac{k\ln{n}}{2}}.
    \end{align*}
    Putting everything together, we get
    $$\Pr(A)\leq ne^{-\frac{n^{1/6}}{4}}+e^{-\frac{k\ln {n}}{2}},$$
    which proves the lemma.
\end{IEEEproof}

The proof of Theorem \ref{Thm: random_cons for BAC} now follows.
\begin{IEEEproof}[Proof of Theorem \ref{Thm: random_cons for BAC}]
We notice first that Construction \ref{random_con} fails if the constructed code is not an $(n,N,k,2\lceil\frac{\sqrt{n}}{s}\rceil)$-BAC with $N\leq n+64(ks\sqrt{n})^{\frac{3}{2}}\ln{n}$.
Let us denote by $A$ the event that Construction \ref{random_con} fails to construct an $(n,N,k,2\lceil\frac{\sqrt{n}}{s}\rceil)$-BAC, and by $B$ the event that the $|\cF|=N-n>64(ks\sqrt{n})^{\frac{3}{2}}\ln{n}$.
Plugging in $p_1=32\sqrt{\frac{(ks)^3}{q}}\ln{n}$ yields
$2p_1n=64(ks\sqrt{n})^{\frac{3}{2}}\ln{n}$. The event $B$ is the event that $|\cF|\geq 2p_1n$.
Thus, the probability that Construction \ref{random_con} fails is given by
\[\Pr(A\cup B)\leq \Pr(A)+\Pr(B).\]

From Lemma \ref{lem:help1}, $\Pr(B)\leq e^{-\frac{p_1n}{3}}$.  Plugging in $p_1=32\sqrt{\frac{(ks)^3}{q}}\ln{n}$ yields $\Pr(B)\leq e^{-n^{3/4}}$. Together with Lemma \ref{lem:help2} we obtain
\[\Pr(A\cup B)\leq e^{-n^{3/4}}+ne^{-\frac{n^{1/6}}{4}}+e^{-\frac{k\ln {n}}{2}}.\]
Since for large enough $n$, this probability is strictly less than $1$, such a code exists and the theorem follows.
\end{IEEEproof}

\bibliographystyle{IEEEtran}
\bibliography{biblio}

\end{document}